\documentclass[journal,web]{ieeecolor}
\usepackage{generic}
\usepackage{cite}
\usepackage{amsmath,amssymb,amsfonts}
\usepackage{algorithmic}
\usepackage{graphicx}
\usepackage{textcomp}

\usepackage{hyperref}
\usepackage{cite}
\usepackage{acronym}
\usepackage{multicol}
\usepackage{cuted}

\newtheorem{theorem}{Theorem}

\newtheorem{definition}{Definition}

\newtheorem{remark}{Remark}

\definecolor{intro_color}{RGB}{150,150,150}

\begin{document}
	\title{Dominant Mixed Feedback Design \\ for Stable Oscillations}
	\author{Weiming Che and Fulvio Forni
		\thanks{W. Che is supported by CSC Cambridge Scholarship. W. Che and F. Forni are with the Department of Engineering, University of Cambridge, CB2 1PZ, UK {\tt\small wc289|f.forni@eng.cam.ac.uk}}}
	
	\maketitle
	
	\begin{abstract}
		We present a design framework that combines positive and negative feedback for
		robust stable oscillations in closed loop. The design is initially based on graphical methods, 
		to guide the selection of the overall strength of the feedback (gain) and of the relative proportion of positive and
		negative feedback (balance). The design is then generalized via linear matrix inequalities. The goal 
		is to guarantee robust oscillations to bounded dynamic uncertainties and to extend the approach
		to passive interconnections.
		The results of the paper provide a first system-theoretic justification to several observations from 
		system biology and neuroscience pointing at mixed feedback as a fundamental enabler for robust oscillations. 
	\end{abstract}
	
	\begin{IEEEkeywords}
	Nonlinear control systems, control system synthesis, closed loop systems, mixed-feedback, linear matrix inequalities.
	\end{IEEEkeywords}
	
	\section{Introduction}
	In control theory we use negative feedback to reduce the error between desired and actual outputs and to {modulate the system closed-loop} sensitivity to uncertainties and disturbances \cite{aastrom2014control}. {The typical goal is to mitigate the effect of nonlinearities and uncertainties, taking advantage of negative feedback to normalize the behavior of a system}. By contrast, positive feedback amplifies the feedback error, leading to instabilities, which often manifest as hysteresis and oscillations \cite{tucker1972history}. This is why positive feedback is less explored in engineering. At the same time, several contributions in system biology and neuroscience point to the fact that nature seems to rely on both positive and negative feedback to generate resilient nonlinear behaviors. Encouraged by this contrast, in this paper we put positive and negative feedback on equal ground. We look at the combination of positive and negative feedback as a key mechanism to achieve rich nonlinear behaviors in closed loop.
	
	{The combination of positive and negative feedback, or \emph{mixed feedback}, is a recurring structure in biological oscillators from molecular level to network level. In their simplest form such oscillators are realized by biochemical reactions such as cell cycles \cite{tsai2008robust}, circadian rhythms \cite{smolen2001modeling} and various gene regulatory circuits \cite{tyson2008biological,mitrophanov2008positive}. At the cellular level, mixed feedback is at the core of the excitability of neurons, which are capable of generating sustained oscillations in terms of repeated spiking and bursting \cite{hodgkin1952quantitative,Drion2015b}. At the network level, more examples of mixed feedback can be found in sensory processing \cite{lee2018combined} and in control of rhythmic movements \cite{marder2001central}.  
	The role of mixed feedback is not only to enable oscillations but also to make these oscillators robust and adaptive, see e.g. \cite{tsai2008robust,ananthasubramaniam2014positive} for biochemical oscillators and \cite{sepulchre2019controlb,sepulchre2019controla,lee2018combined} for neural oscillators.
	In engineering, mixed feedback loops can be traced in the biochemical oscillators of synthetic biology \cite{Novak2008,o2012modeling}, in neuromorphic circuits \cite{ribar2019neuromodulation}, and in robotic locomotion \cite{kimura1999realization} (see also \cite{ijspeert2007swimming,ijspeert2008central} for an introduction to the role of endogenous oscillators in robotic locomotion).
	 The use of mixed feedback in these domains
	often relies on a few design tools, among which we have} harmonic balance methods \cite{vander1968multiple,iwasaki2008multivariable}, and specific methods for relaxation oscillations \cite{aastrom1995oscillations}. 
	
	{The discussion and the design methods developed in our paper} are
	strongly influenced by the system-theoretic characterization of neuronal excitability in \cite{Drion2015b,sepulchre2019controla,sepulchre2019controlb}. These papers point at mixed feedback as a fundamental enabler for nonlinear behaviors. We build upon this view. Our goal is to develop {a systematic design based on mixed feedback for stable oscillations in closed loop}. As in the preliminary results of \cite{che2021tunable}, we investigate the simplest realization of a \emph{mixed feedback controller} given by the parallel interconnection of two stable first order linear networks, whose action is combined into a sigmoidal nonlinearity. 
 The mixed feedback controller is regulated by the \emph{gain}, $k$, which controls the overall feedback strength, and by the \emph{balance}, $\beta$, which controls the relative strength between positive and negative feedback actions (Section \ref{Sec:mixed feedback controller}). The sigmoidal saturation is the most common nonlinearity in physical systems, representing the finite voltage supplied to an electrical circuit and the limited force/torque supplied to a mechanical system. The problem of mixed feedback design is thus the problem of finding suitable values for the control parameters to guarantee stable oscillations in closed loop. 
	
	The analysis and the design approaches proposed in this paper are based on dominance theory and differential dissipativity \cite{forni2018differential,miranda2018analysis}. {We use dominance theory to determine whether a closed-loop high-dimensional mixed-feedback system has a low dimensional ``dominant’’ behavior. We take advantage of the fact that} the attractors of a $2$-dominant system correspond to the attractors of a planar system ($2$-dimensional). This means that the Poincar\'e-Bendixon theorem can be used on a 2-dominant system to certify oscillations, even if the system has a large dimension. {Later in the paper we use differential dissipativity to extend our results to open nonlinear systems. The goal is to develop a robust design framework for mixed-feedback oscillators, which mimics classical linear robust theory, and to propose a passivity-based approach for interconnections}.
	
	{The novel contributions of the paper are summarized below:	
	(i) via root locus analysis (Section \ref{Section:Root_Locus}), we show why fast positive feedback and slow negative feedback are needed for 2-dominance of the mixed feedback closed loop, thus to support oscillations. Our findings, based on necessary conditions for $2$-dominance, agree with and support several observations from biology, where positive feedback is typically fast. 
	(ii)
	Section \ref{Section:Sufficiency} combines circle criterion for dominance \cite{miranda2018analysis} and local stability analysis to derive sufficient conditions for stable oscillations. This leads to basic rules for the selection of the gain and the balance of the mixed feedback controller. The discussion in Section \ref{Section:Sufficiency} combines and extends the early results of \cite{che2021tunable,che2021shaping}. To the best of our knowledge, the analysis of robustness in Section \ref{Section:Example1_robustness} is new and shows how classical Nyquist arguments can be used to provide a graphical quantification of robustness for mixed-feedback oscillators.
	(iii)
	 From Section \ref{Section:p-dissipativity} we generalize our approach by using linear matrix inequalities (LMIs) for oscillator design. LMIs for dominance analysis were presented in \cite{forni2018differential,miranda2018analysis} and we extend their use for design purposes, taking advantage of the particular structure of the mixed-feedback controller.
	(iv)
	Section \ref{Sec:Robust_and_Interconnection} provides a state-feedback design to guarantee robustness of oscillations to bounded uncertainties. We also tackle passive interconnections. Our design shows strong similarities with classical approaches for stabilization of equilibria. However, the difference is that the mixed-feedback controller ``destabilizes” the closed-loop system while enforcing a specific degree of dominance. This ensures oscillations and is achieved by a new constraint on the inertia of the matrix $Y$ in \eqref{eq:LMI_2_dominant}, \eqref{eq:LMI_unstable_origin}, \eqref{eq:Robustness_LMI}, and
\eqref{eq:Gerenal_passive_interconnection_LMI}. 

	From a more general perspective, the paper shows that dominance theory and differential dissipativity can be leveraged for the design of closed-loop oscillators in mixed-feedback setting. Our approach leads to constructive conditions for synthesis based on graphical methods and LMIs. This allows for the design of oscillators that are robust to prescribed level of uncertainties.  We believe that our paper contributes to clarifying the role of mixed feedback in enabling robust oscillations, as observed by other scientific communities (e.g., system biology, neuroscience). To the best of our knowledge, these observations are of qualitative nature, based on extensive simulations and on phase plane analysis (on reduced system dynamics). In contrast, our paper deals with large dimensional systems and provide quantitative conditions for robust design of mixed-feedback oscillators.}
		
	\section{Mixed Feedback Controller}
	\label{Sec:mixed feedback controller}
	
	{To focus our analysis on the interplay between positive and negative feedback for the generation
	of endogenous oscillations, we consider a  closed loop with minimal nonlinearities. 
	The \emph{mixed feedback closed loop} is illustrated in Figure \ref{fig:Block}. It consists of linear plant dynamics, $\mathcal{P}(s)$, and mixed feedback controller, given by positive and negative passive linear networks, $\mathcal{C}_p(s)$ and $\mathcal{C}_n(s)$, and by a static saturation nonlinearity $\varphi$.}
		We assume that $\mathcal{P}(s)$ is an asymptotically stable and strictly proper single input single output transfer function. $\mathcal{C}_p(s)$ and $\mathcal{C}_n(s)$ are first order lags where $\tau_p$ and $\tau_n$ are the corresponding time constants
	\begin{equation}\label{eq:Mixed_Feedback_Controller}
		\mathcal{C}_p(s)=\frac{1}{\tau_p s+1} \ , \quad  \mathcal{C}_n(s)=\frac{1}{\tau_n s+1} \ .
	\end{equation}
	{
	These assumptions are motivated by reasons of simplicity and can be relaxed, 
	as clarified in later sections. } 
	
	\begin{figure}[!h]
		\centering
		\includegraphics[width=0.4\textwidth]{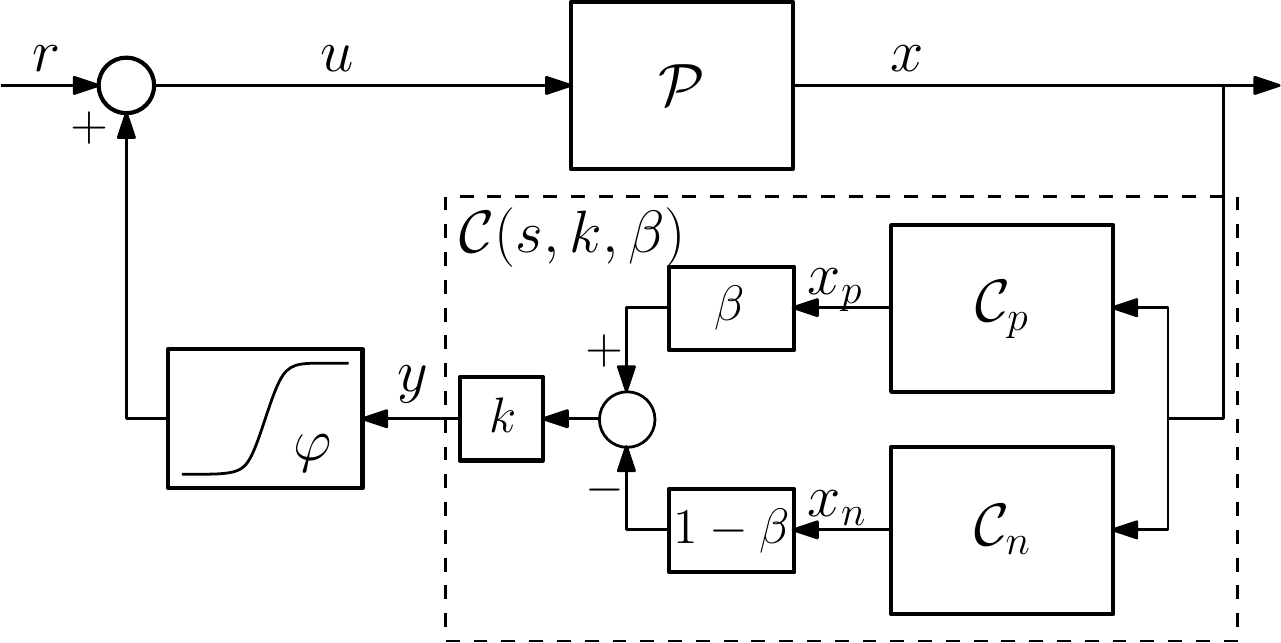}
		\caption{Block diagram of the mixed feedback closed loop system.}
		\label{fig:Block}
	\end{figure}
	
	Here and in what follows we assume that {$\tau_p \neq \tau_n$ and that both time constants are} \emph{slower} than any of the plant time constant. 
	{The motivation for these assumptions is to show that the generation of oscillations via mixed feedback works also for non-resonant plants, whose dynamics typically
	have a fast decay rate. In practice, we are free to choose the parameters of our controller, thus this assumption does not pose a strong limitation.}
	
	The action of the mixed feedback controller is regulated by the parameters $k\geq 0$ and $0 \leq \beta \leq 1$, where $k$ regulates the overall feedback \emph{gain} while $\beta$ regulates the \emph{balance} between positive and negative feedback. The {linear part of} the mixed feedback controller {has transfer function}:\begin{equation}\label{eq:mixed_feedback_controller}
		\mathcal{C}(s,k,\beta)=\frac{k\Big(\big(\beta(\tau_n+\tau_p)-\tau_p\big)s+2\beta-1\Big)}{(\tau_ps+1)(\tau_ns+1)}
	\end{equation}
	The static nonlinearity $\varphi$ {represents a  monotone, slope restricted, bounded actuation stage}. In this paper $\varphi$ {is a differentiable, sigmoidal function with slope $0 \leq \partial \varphi \leq 1$ (here and in what follows $\partial \varphi$ denotes $\frac{\partial \varphi(y) }{\partial y}$).} We also assume that $|\varphi| \leq M$, for some finite number $M$. This guarantees the \emph{boundedness of the closed loop trajectories} for any selection of the feedback parameters $k$ and $\beta$. 
	For simplicity, each simulation in this paper will adopt $\varphi=\tanh$.
	
	The closed loop can be represented as a Lure system as shown in Figure \ref{fig:Lure_theory}, where 
	\begin{equation}\label{eq:linear_tf}
		G(s,k,\beta)=-\mathcal{C}(s,k,\beta)\mathcal{P}(s)
	\end{equation}
	
	\section{Dominance Theory} 
	\label{Sec:Dominance_Th}
	{Our analysis (and design) of the mixed feedback controller uses tools from dominance theory and differential dissipativity \cite{forni2018differential,miranda2018analysis}, which show strong contact points with the theory of monotone systems with respect to high rank cones \cite{Sanchez2009,Sanchez2010} and with contributions that extend the Poincar\'e-Bendixon theorem to large dimensional systems \cite{Smith1980}, \cite{Smith1986}. We refer the reader to \cite[Section V]{forni2018differential} for a detailed comparison with the literature.}
	
	{Dominance theory extends several classical tools of linear control theory to the analysis
	of nonlinear systems with complex attractors like limit cycles. Through conic constraints, dominance theory shows that the dynamics of a nonlinear system
	can be partitioned} into ``fast fading" dynamics and ``slow dominant" dynamics {\cite[Theorem 1]{forni2018differential}}. The dominant dynamics, typically of small dimension, drive the asymptotic behavior of the nonlinear system {\cite[Theorem 2]{forni2018differential}}. {These dynamics are not necessarily stable and may generate multi-stable or oscillatory behaviors, as clarified in Theorem \ref{th:p-attractor}, below. 
	In what follows we summarize the main results of the theory, which are later used in the paper. For a detailed introduction we refer
	the reader to \cite{forni2018differential} and \cite{miranda2018analysis}}. 
	
	{Consider a nonlinear system of the form}
	\begin{equation} \label{theorem:general_nonlinear}
		\dot{x}=f(x)\qquad x\in\mathbb{R}^n \ ,
	\end{equation}
	{where $f$ is a continuously differentiable vector field.
	Dominance theory makes use of the system linearization along arbitrary trajectories, characterized by the prolonged system} 
	\begin{equation}\label{theorem:prolonged_sys}
		\begin{cases}
			\dot{x}=f(x)\\
			\delta\dot{x}=\partial f(x)\delta x
		\end{cases}\quad (x,\delta x)\in \mathbb{R}^n\times\mathbb{R}^n \ .
	\end{equation}
	Here $\partial f(x)$ represents the Jacobian of $f$ computed at $x$. {$\delta x \in \mathbb{R}^n$ represents a generic tangent vector at $x\in \mathbb{R}^n$. As usual, $\dot{x}$ and $\delta \dot{x}$ are short notations for $\frac{d}{dt} x$ and $\frac{d}{dt} \delta x$, respectively. 
	Along any generic trajectory $x(\cdot)$ of \eqref{theorem:general_nonlinear}, the related sub-trajectory $\delta x(\cdot)$ of \eqref{theorem:prolonged_sys} can be interpreted as a small perturbation / infinitesimal variation, as clarified in \cite{Forni2014}.} 
	
	\begin{definition}{\cite[Definition 2]{forni2018differential}}\label{de:p-dominance}
		The nonlinear system \eqref{theorem:general_nonlinear} is $p$-\emph{dominant with rate} $\lambda\geq0$ if and only if there exist a symmetric matrix $P$ with inertia $(p,0,n-p)${\footnote{A symmetric matrix $P$ with inertia $(p,0,n-p)$ has $p$ negative eigenvalues and $n-p$ positive eigenvalues.}} and $\varepsilon\geq0$ such that the prolonged system \eqref{theorem:prolonged_sys} satisfies the conic constraint:
		\begin{equation}\label{eq:dominance_LMI}
			\begin{bmatrix}
				\delta\dot{x}\\
				\delta x
			\end{bmatrix}^T \begin{bmatrix}
				0&P\\P&2\lambda P+\varepsilon I
			\end{bmatrix}\begin{bmatrix}
				\delta\dot{x}\\
				\delta x
			\end{bmatrix}\leq 0
		\end{equation}
		for all $(x,\delta x) \in \mathbb{R}^{2n}$. The property is strict if $\varepsilon>0$.$\hfill\lrcorner$
	\end{definition}
	{
	The conic constraint \eqref{eq:dominance_LMI} can be equivalently formulated as the following matrix inequality
	$$
	(\partial f(x)^T + \lambda I) P + P (\partial f(x) + \lambda I)  \leq -\varepsilon I \qquad \forall x \in \mathbb{R}^n \ .
	$$
	This inequality shows how $p$-dominance} guarantees that $n-p$ eigenvalues of the Jacobian matrix $\partial f(x)$ lie to the left of $-\lambda$ while the remaining $p$ eigenvalues lie to the right of $-\lambda$, for each $x$ {\cite[Theorem 3]{forni2018differential}}. This splitting, {uniform with respect to $x$, is a \emph{necessary condition} for dominance.} 
	
	Dominance can be also characterized in the frequency domain, for systems that have a Lure representation 
	 {as in Figure \ref{fig:Lure_theory}. The following conditions are particularly useful for design purposes.}	
	
	\begin{theorem}{\cite[Corollary 4.5]{miranda2018analysis}} \label{th:circle_cirteria}
		Consider the Lure feedback system in Figure \ref{fig:Lure_theory} given by the negative feedback interconnection of the linear system $G(s)$ and the static nonlinearity $\varphi$, satisfying the sector condition {$0 \leq \partial\varphi \leq K$}. The closed system is strictly $p$-dominant with rate $\lambda$ if
		\begin{enumerate}
			\item the real part of all the poles of $G(s)$ is not $-\lambda$;
			\item the shifted transfer function $G(s-\lambda)$ has $p$ unstable poles;
			\item the Nyquist plot of $G(s-\lambda)$ lies to the right of the vertical line passing through the point $-1/K$ on the Nyquist plane.$\hfill\lrcorner$
		\end{enumerate} 
	\end{theorem}
	
	\begin{figure}[htbpt]
		\centering
		\includegraphics[width=0.28\textwidth]{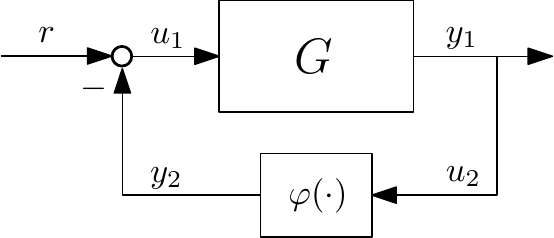}
		\caption{Lure feedback system}
		\label{fig:Lure_theory} \label{fig:Fixed_point_Block}
	\end{figure}	
	
	We are particularly interested in $p$-dominant systems with a small degree $p \leq 2$. 
	{A small degree guarantees that the nonlinear system possesses a simple attractor
	 (not necessarily an equilibrium), as clarified below.}
	\begin{theorem}{\cite[Corollary 1]{forni2018differential}}\label{th:p-attractor}
		Consider a $p$-dominant system $\dot{x}=f(x)$, $x\in\mathbb{R}^n$, with dominant rate $\lambda\ge0$. Every bounded trajectory of the system asymptotically converges to
		\begin{itemize}
			\item a unique fixed point if $p=0$;
			\item a fixed point if $p=1$;
			\item a simple attractor if $p=2$, that is, a fixed point, a set of
			fixed points and connecting arcs, or a limit cycle.$\hfill\lrcorner$
		\end{itemize} 
	\end{theorem}
	
	Theorem \ref{th:p-attractor} shows how dominance theory can be used to shape the behavior of the mixed-feedback closed loop.
	{The first step is to \emph{find the set of parameters that guarantees a low degree of dominance}.} For $p=0$ the mixed-feedback closed loop is contractive. That is, the system relaxes to a single steady-state behavior for any given reference $r$. Likewise, for $p = 2$ the system behaves like a planar system.
	{Then, the second step is to enforce oscillations by \emph{deriving the subset of the control parameters that guarantees $2$-dominance and instability of closed-loop equilibria} at the same time. For these parameters, Theorem \ref{th:p-attractor} guarantees that
	a system with bounded trajectories must have a limit cycle.
	In what follows we will use the terminology of \emph{stable oscillations} to denote
	the existence of a stable limit cycle in the state space of the nonlinear system \eqref{theorem:general_nonlinear}.}

	\section{Fast positive / slow negative mixed \\ feedback for $2$-dominant systems}
	\label{Section:Root_Locus}
	{The connection between $2$-dominance and oscillations informs the selection of the
	mixed-feedback control parameters. The goal is to  
	identify parameter ranges that are compatible with $2$-dominance, thus with
	oscillations. We focus on the two time constants 
	$\tau_p$ and $\tau_n$, on the balance $\beta$, and on the gain $k$.  
	Our analysis shows that the combination of \emph{fast positive feedback and slow negative feedback}, 
	$\tau_p < \tau_n$, is a necessary condition for $2$-dominance. This agrees with several observations 
	from biology and neuroscience, which show how fast positive / slow negative feedback is a source of robust oscillations.} 
	
	{We focus on the eigenvalues of the Jacobian of the closed-loop system, computed at a generic point $x$
	of the system state space. The goal is to identify the set of parameters that guarantees a uniform splitting of these
	eigenvalues, for all $x$ (necessary condition for dominance). For fixed time constants $\tau_p$ and $\tau_n$, the eigenvalues of the Jacobian
	of the linearized mixed feedback closed loop are necessarily contained within the root loci of $G(s,1,\beta)$, for $\beta \in [0,1]$. 
	This follows from the fact that $\varphi$ is the only nonlinearity of the system and its derivative satisfies $0 \leq \partial\varphi  \leq 1$. 
	For dominance, we look for conditions that guarantee that the root loci present a splitting of the roots into two groups, respectively to the left and to the right of some dominance rate $-\lambda$. We are also interested in those situations where the closed loop loses stability, which is a necessary condition to generate oscillations.}
	
		\begin{figure}[t]
		\centering
		\includegraphics[width=0.37\textwidth]{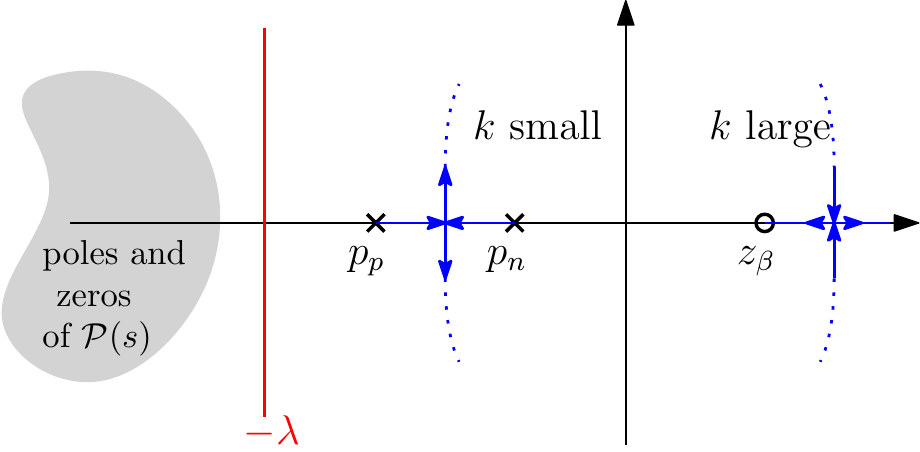}
		\vspace{-1mm}
		\caption{Root locus for $\tau_p < \tau_n$ and $\beta^* \!<\! \beta \!<\! 1$.}
		\label{fig:Rootlocus1}
	\end{figure}
	\begin{figure}[t]
		\centering
		\includegraphics[width=0.37\textwidth]{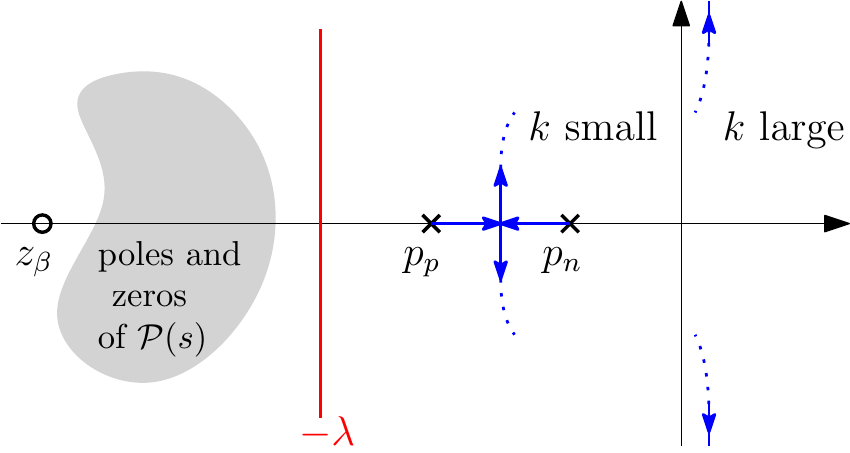}
		\vspace{-1mm}
		\caption{Root locus for $\tau_p < \tau_n$ and $0 \!<\! \beta \!<\! \beta^*$.}
		\label{fig:Rootlocus2}
	\end{figure}
	
		\begin{figure}[t]
		\centering
		\includegraphics[width=0.37\textwidth]{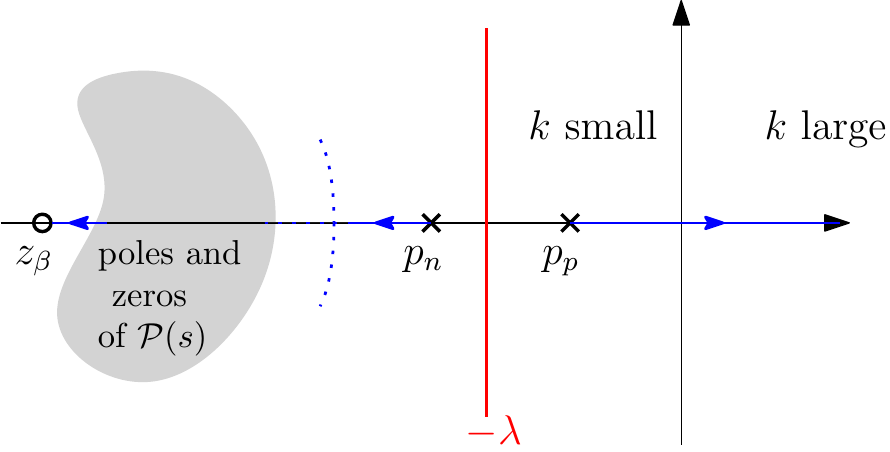}
		\vspace{-1mm}
		\caption{Root locus for $\tau_p > \tau_n$ and $\beta^* \!<\! \beta \!<\! 1$.}
		\label{fig:Rootlocus3}
	\end{figure}
	
	\begin{figure}[t]
		\centering
		\includegraphics[width=0.37\textwidth]{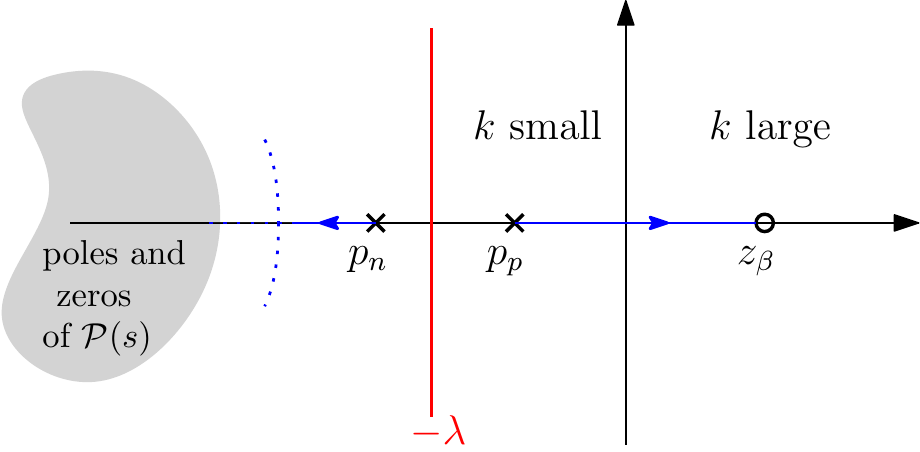}
		\vspace{-1mm}
		\caption{Root locus for $\tau_p > \tau_n$ and $0 \!<\! \beta \!<\! \beta^*$.}
		\label{fig:Rootlocus4}
	\end{figure}
	
	We will consider two arrangements of the relative time scale of positive and negative feedback:
	(i)~fast positive/slow negative feedback: $\tau_p\!<\!\tau_n$ and
	(ii)~slow positive/fast negative feedback: $\tau_p\!>\!\tau_n$.
	The following observations are instrumental to the graphical {analysis summarized by} Figures \ref{fig:Rootlocus1}-\ref{fig:Rootlocus4}.
	The mixed feedback controller \eqref{eq:mixed_feedback_controller} has two poles and a zero
	\begin{equation}\label{eq:TF_poles_and_zero}
		\begin{cases}
			p_p&=-\frac{1}{\tau_p}\\
			p_n&=-\frac{1}{\tau_n}\\
			z_\beta&=-\frac{2\beta-1}{\beta(\tau_p+\tau_n)-\tau_p}
		\end{cases}.
	\end{equation}
	The position of the zero $z_\beta$ is a function of the balance parameter $\beta$. 
	As $\beta$ moves within the interval $[0,1]$, $z_\beta$ explores the set $(-\infty,\min(p_p,p_n))\cup(\max(p_p,p_n),\infty)$.
	{For $\beta = 0$, $z_\beta$ corresponds to $p_p$ and for $\beta = 1$, $z_\beta$ corresponds to $p_n$. 
	$z_\beta$ crosses $0$ for $\beta = 0.5$ and it approaches $\pm\infty$ as $\beta$ approaches the critical value $\beta^*=\frac{\tau_p}{\tau_p+\tau_n}$.
	Thus, 
	\begin{itemize}
	\item for $\tau_p < \tau_n$, $z_\beta$ remains to the right of the slowest pole $p_n$ for $\beta^* < \beta < 1$ ($\beta^* < 0.5$); and
	\item for $\tau_p > \tau_n$, $z_\beta$ remains to the right of the slowest pole $p_p$ for $0 < \beta < \beta^*$ ($0.5 < \beta^*$).
	\end{itemize}}

	{For the \emph{fast positive/slow negative feedback} case, $\tau_p < \tau_n$,  taking} 
%
%
	$\beta^*< \beta < 1$ (strong positive feedback), the closed loop system admits a root locus of positive feedback convention, as shown in Figure \ref{fig:Rootlocus1}. By the assumptions on the mixed feedback closed loop in Section \ref{Sec:mixed feedback controller}, the poles of the mixed feedback controller lie to the right of the poles of the plant. This means that the open loop poles can be split into transient (plant) and dominant (controller). Furthermore, the position of the zero $z_\beta$ guarantees that this splitting persists for a sizable interval of gains $0 \leq k \leq k^*$ ($k^*$ could be $\infty$ for plant with small relative degree). For small $k \geq 0$, all the poles of the linearized system are stable. In this case, the system is compatible with $0$-dominance for $\lambda=0$. The equilibrium at 0 remains stable and no oscillations occur. For all $0 \leq k < k^*$, the system is also compatible with $2$-dominance with rate $\lambda > 1/\tau_p$ (but such that $-\lambda$ remains to the right of the poles of $\mathcal{P}(s)$). Furthermore, when  $z_\beta$ has positive real part, $k$ large enough guarantees that the poles of the linearized system cross the imaginary axis. The origin of the closed loop becomes unstable and nonlinear behaviors like multi-stability and oscillations may appear.
	
	When $0 < \beta<\beta^*$, the closed system admits a negative feedback root locus. This case is more complicated. A splitting compatible with $2$-dominance is preserved if $z_\beta$ belongs to the left of $-\lambda$, as shown in Figure \ref{fig:Rootlocus2}. Furthermore, if $\beta$ is sufficiently close to $\beta^*$, the intersection point of the root locus asymptotes belongs to right-half plane. This guarantees that the origin becomes unstable for large $k$. We can draw the following conclusions:
	\begin{itemize}
		\item The system can be $0$-dominant for $k$ sufficiently small.
		\item For $\beta^* < \beta < 1$ and {for a suitable subset of $0 < \beta < \beta^*$}, the mixed feedback system can be $2$-dominant. Nonlinear behaviors like multi-stability and oscillations may emerge when $k$ is sufficiently large.
	\end{itemize} 
		
	{For the \emph{fast negative/slow positive feedback} case, $\tau_p > \tau_n$,} 
	the generic root locus plots {for strong, $\beta^* < \beta < 1$, and weak, $0 < \beta<\beta^*$, positive feedback are shown in Figures \ref{fig:Rootlocus4} and \ref{fig:Rootlocus3}}, respectively. $0$- and $2$- dominance may hold for small $k$, however, under such time scale arrangement, $2$-dominance does not hold for large $k$. This is due to the fact that $p_n$ moves to the left for increasing $k$. {No selection of the rate $\lambda$ guarantees a splitting with two right poles}.
	
	Interestingly, the configuration of {the root loci of the closed-loop system are compatible with $1$-dominance. This means that} the system may exhibit multiple equilibria (but no oscillations).	
	We can draw the following conclusions:
	\begin{itemize}
		\item The system can be $0$-dominant for $k$ sufficiently small.
		\item For both weak and strong positive feedback, the system can be $1$-dominant {for selected $-\lambda$ between $p_p$ and $p_n$}. In this case, no oscillation can take place. Multiple equilibria may appear for $k$ sufficiently large.
	\end{itemize}
	
	Since we are interested in 2-dominance and oscillations, in what follows we will focus on the mixed feedback controller with fast positive feedback and slow negative feedback, leaving aside the case of slow positive feedback and fast negative feedback. The latter is not compatible with $2$-dominance.

	\section{$2$-dominant mixed feedback \\ design for oscillations}
	\label{Section:Sufficiency}
	
	\subsection{Control design}\label{subsec:frequency_domain_analysis}
	{The previous section provides insights into the time-scale structure
	of mixed feedback. The arrangement into fast positive and slow negative feedback 
	is compatible with $2$-dominance, thus can be leveraged to achieve oscillations. 
	In this section we develop a design procedure to find parameter ranges for gain and balance, $k$ and $\beta$,
	that \emph{guarantee} endogenous oscillations for the closed-loop system (i.e. the existence of a stable limit cycle).

	The design procedure has two steps. The first identifies gain and balance ranges 
	that guarantee $2$-dominance. This is achieved by leveraging Nyquist diagrams and the circle criterion 
	for $2$-dominance (Theorem \ref{th:circle_cirteria})}, as summarized by Theorem \ref{th:0-dominance} and \ref{th:2-dominance}. 
	 {
	The second step  combines the $2$-dominance property with a 
	detailed analysis of the stability of the closed-loop equilibria. This allows us to use
	Theorem \ref{th:p-attractor} to certify oscillations. }

	\begin{theorem}\label{th:0-dominance}
		For any constant reference $r$ and any $\beta\in[0,1]$, the mixed feedback system in Figure \ref{fig:Block} is 0-dominant with rate $\lambda=0$ 
		for any gain $0\leq k<k_0$, where
		\begin{equation} \label{eq:k0}
			k_0 = 
			\begin{cases}
				\infty \! & \! \mbox{if } \min\nolimits\limits_{\omega} \Re(G(j \omega ,\!1,\!\beta)) \!\geq\! 0 \\
				- \dfrac{1}{\min\nolimits\limits_\omega \Re(G(j\omega ,\!1,\!\beta))} \!&\!  \mbox{otherwise. } 	
			\end{cases} \vspace{-6mm}
		\end{equation}$\hfill\lrcorner$
		\vspace{2mm}
	\end{theorem}
	
	\begin{proof}
		According to Theorem \ref{th:circle_cirteria}, the mixed feedback closed loop is $0$-dominant if the Nyquist plot of the linear system $G(s,k,\beta)$ lies to the right hand side of the line $-1$ in the complex plane as shown in Figure \ref{fig:circle_criteria_zero}. 
		Note $G(s,k,\beta)=kG(s,1,\beta)$, i.e. $k$ only scales the magnitude of $G(s,k,\beta)$. Hence the condition on the Nyquist plot of $G(s,k,\beta)$ is verified whenever $0 \leq k < k_0$, by construction.
	\end{proof}
	
	\begin{figure}[!h]
		\centering
		\includegraphics[width=0.24\textwidth]{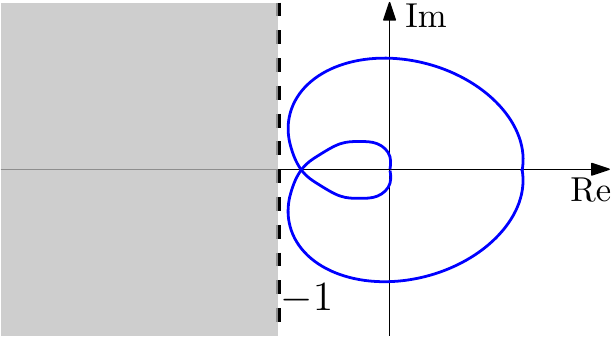}
		\vspace{-1mm}
		\caption{An illustration of the circle criteria for dominance.}
		\vspace{-1mm}
		\label{fig:circle_criteria_zero}
	\end{figure}
	
	\begin{theorem}\label{th:2-dominance}
		Consider a rate $\lambda$ for which the 
		transfer function $G(s-\lambda,1,\beta)$ has two unstable poles. Then, 
		for any constant reference $r$ and any $\beta\in[0,1]$, 
		the mixed feedback system in Figure \ref{fig:Block} is 2-dominant with rate $\lambda$ 
		for any  gain $0 \leq k < k_2 $, where
		\begin{equation}\label{eq:k2} 
			\!\!\!k_2 \!=\! 
			\left\{\begin{array}{ll}
				\!\!\!\!\infty \!\! & \!\! \mbox{if } \!\min\nolimits\limits_{\omega} \Re(G(j \omega\!-\!\lambda ,\!1,\!\beta)) \!\geq\! 0 \\
				\!\!\!- \frac{1}{\min\nolimits\limits_\omega \Re(G(j\omega-\lambda ,\!1,\!\beta))} \!\!&\!\!  \mbox{otherwise. }	
			\end{array}\right. 
		\end{equation}
		\vspace{-5mm}
		$\hfill\lrcorner$
	\end{theorem}
	\begin{proof}
		The proof argument is similar to the one of Theorem \ref{th:0-dominance}, using the shifted transfer function $G(s-\lambda, k,\beta)$.
	\end{proof}
	
	For any given $\beta$, Theorems \ref{th:0-dominance} and \ref{th:2-dominance} state conditions on the gain $k$ for $0$-dominance and $2$-dominance, respectively. This agrees with the root locus in Figure \ref{fig:Rootlocus1}. Thus there is a range of gains $k$ for which the system is both $0$-dominant and $2$-dominant. In this case, the system behavior will satisfy the most restrictive condition, namely $0$-dominance. The root locus analysis also suggests that $k_2$ is greater than $k_0$, since $2$-dominance is compatible with unstable {equilibria}. As a result, {to explore the oscillatory regime, we will look at the range of gains $k_0 < k < k_2$.}

	When the assumptions of Theorem \ref{th:0-dominance} are satisfied the closed-loop trajectories 
	asymptotically converge to the unique equilibrium that is compatible with the constant reference $r$. 
	In contrast, Theorem \ref{th:2-dominance} guarantees that {the closed-loop system has simple attractors. 
	Thus, in the spirit of the Poincar\'e-Bendixson theorem \cite{Hirsch1974}, oscillations will appear within forward invariant regions that 
	do not contain equilibria. Indeed, in a $2$-dominant system with bounded trajectories, stable oscillations are guaranteed 
	to appear (i.e. a stable limit cycle exists) if all equilibria are unstable.}
	
	{Denote by $u_1$, $y_1$ the input-output pair of $G(s,k,\beta)$ and by $u_2$, $y_2$ the input-output pair of $\varphi(\cdot)$, as  shown in Figure \ref{fig:Fixed_point_Block}. At each equilibrium the closed-loop system satisfies the following equations:}
	
	\begin{equation*}
		\begin{cases}
			y_1=G(0,k,\beta)u_1\\
			y_2=\varphi(u_2)\\
		\end{cases}
		\quad
		\begin{cases}
			u_1=-y_2+r\\
			u_2=y_1 \ ,
		\end{cases}
	\end{equation*}
	{that is,}
	\begin{equation}\label{eq:fixed_point1}
		-\frac{y_1}{G(0,k,\beta)}+r=\varphi(y_1).
	\end{equation} 
	{Using \eqref{eq:linear_tf}, $G(0,k,\beta)=-k(2\beta-1)\mathcal{P}(0)$,
	thus \eqref{eq:fixed_point1} reads}
	\begin{equation}\label{eq:fixed_point2}
		\frac{y_1}{k\mathcal{P}(0)(2\beta-1)}+r=\varphi(y_1).
	\end{equation}
	The slope $\frac{1}{k\mathcal{P}(0)(2\beta-1)}$ determines the number of equilibria. 

		
	For example, consider $\varphi =\tanh$ (the analysis below can be {easily adapted to any monotone, slope-restricted, bounded} static nonlinearity). For $\beta\in[0,0.5]$ we have $\frac{1}{k\mathcal{P}(0)(2\beta-1)}\leq0$, thus there is only one equilibrium. By contrast, when $\beta\in(0.5,1]$, the system may have multiple equilibria as $k\mathcal{P}(0)$ and $r$ varies.
 For $r=0$, there are four possible configurations for different $k\mathcal{P}(0)$ values, as shown in Figure \ref{fig:fixed_point_stability}. The situation for case of $r \neq 0$ can be simply deduced by shifting the straight line vertically by $r$.
	
	Given any balance $\beta$, the slope $\frac{y_1}{k\mathcal{P}(0)(2\beta-1)}$ converges to zero as 
	$k$ increases. Two more equilibria will appear for $\frac{1}{k\mathcal{P}(0)(2\beta-1)}\in(0,1)$, i.e. $k\mathcal{P}(0)>\frac{1}{2\beta-1}$. The stability of the equilibria is characterized
	by the gray area between dashed lines in Figure \ref{fig:fixed_point_stability}, which distinguishes the stable (outside the gray area) and unstable (inside the gray area) linearization of the mixed feedback closed loop. The gray area is derived via Nyquist criterion.
	
	The lack of gray area in Figure \ref{fig:fixed_point_stability}.a indicates that the equilibrium point at $0$ is always stable. This happens when $k$ is small. As $k$ and/or $\beta$ increase, Figures \ref{fig:fixed_point_stability}.b and \ref{fig:fixed_point_stability}.c show that all the equilibria become unstable. For these cases, unstable equilibria combined to boundedness of trajectories and $2$-dominance guarantee that oscillations will occur. A further increase of the gain stabilizes two of the three equilibria, as shown in Figure \ref{fig:fixed_point_stability}.d. In such case, oscillations may disappear in favor of (or coexist with) a bistable behavior. {This occurs for very large positive feedback.}
	
	\begin{figure}[htbp]
		\centering
		\includegraphics[width=1\columnwidth]{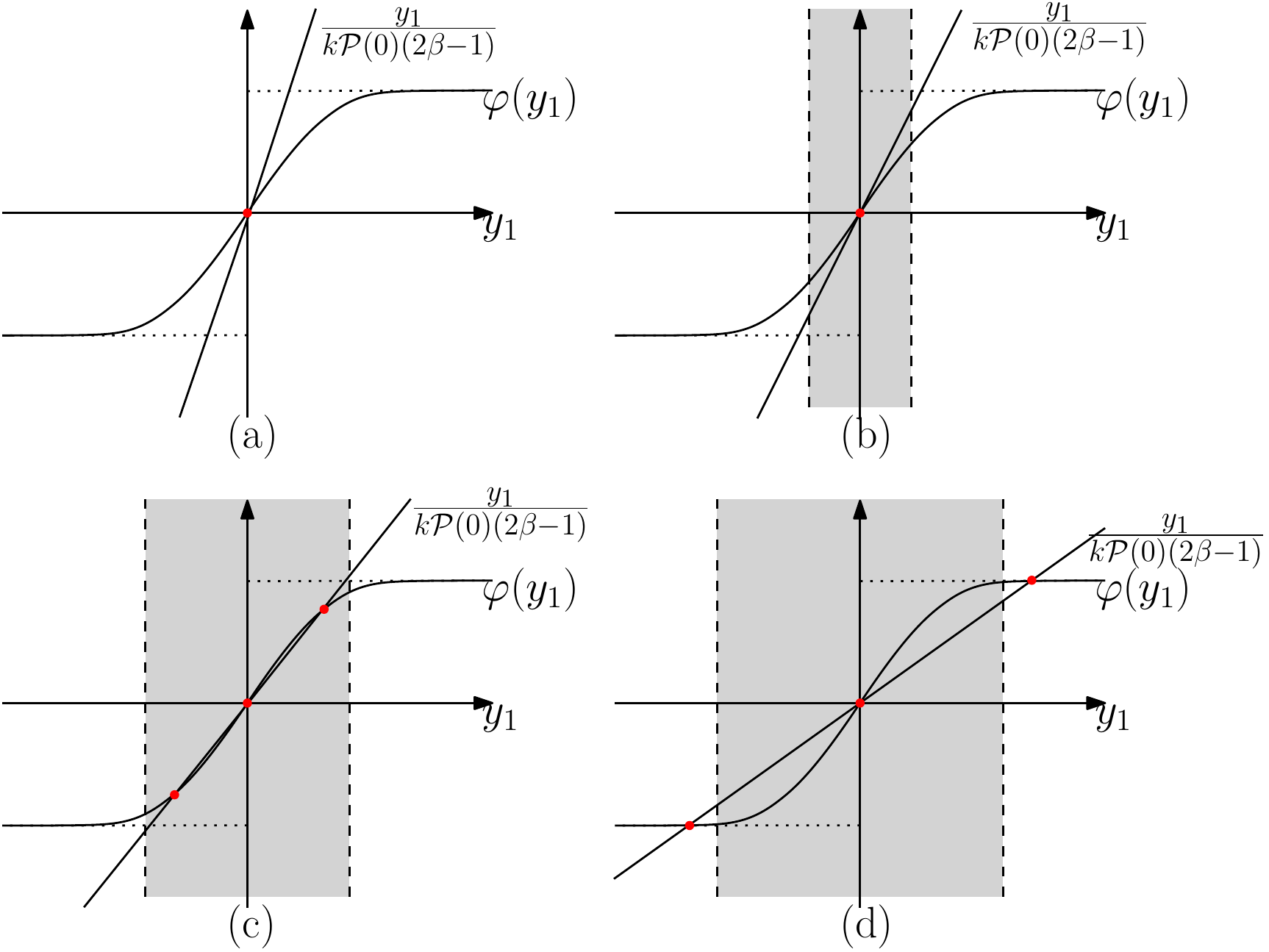}
		\caption{Stability of the equilibria for input reference $r=0$.}
		\label{fig:fixed_point_stability}
	\end{figure}
	
	\subsection{Robust oscillations to plant uncertainties}
	\label{Section:Example1}
	\label{Section:Example1_robustness}
	For illustration, we take $\mathcal{P}(s)$ as a simple first order lag $\mathcal{P}(s)=\frac{1}{\tau_ls+1}$ and numerically compute the ranges of $(k,\beta)$ that lead to different dominant properties and different {stability properties of the closed-loop equilibria}. We take $\beta\in[0,1]$ and $k\in(0.1,1000)$, and we consider different time scale arrangements, reflecting strong and weak time-scale separation: $\tau_l=0.01$, $\tau_p=0.1$, $\tau_n=1$ in Figure \ref{fig:dominance_map}.a, and $\tau_l=0.01$, $\tau_p=0.1$, $\tau_n=0.3$, in Figure \ref{fig:dominance_map}.b. We set $\lambda=50$, roughly in the middle of left most two poles. The input reference $r$ is set to $0$.
	
	\begin{figure}[htbp]
		\centering
		\includegraphics[width=0.9\columnwidth]{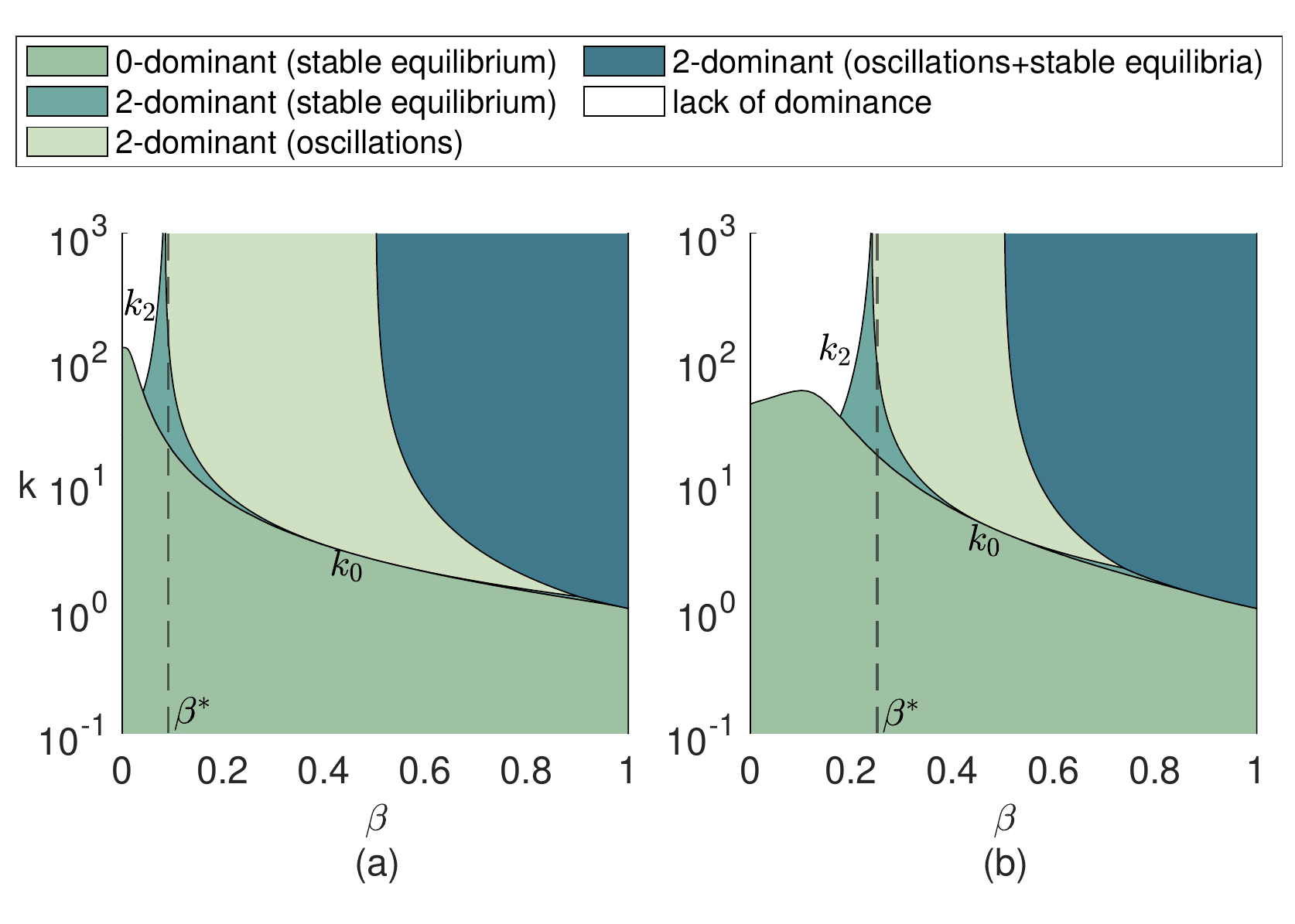} \vspace*{-3mm}
		\caption{{Dominance map for $r=0$ and rate $\lambda=50$. \textbf{(a)}: $\tau_l=0.01$, $\tau_p=0.1$, $\tau_n=1$; \textbf{(b)}: $\tau_l=0.01$, $\tau_p=0.1$, $\tau_n=0.3$.}}
		\label{fig:dominance_map}
	\end{figure}
	
	As shown in Figure \ref{fig:dominance_map}, the parametric $(k,\beta)$ plane is divided into five regions. The lines $k_0$ and $k_2$ mark the upper bound of $0$-dominance and $2$-dominance respectively, as discussed in section \ref{Section:Root_Locus}. The white region above them, denoted as ``lack of dominance'', is where circle criteria for dominance is not satisfied for both $0$-dominance and $2$-dominance. The system is $0$-dominant and globally stable for $k<k_0$ for all $\beta\in[0,1]$. The relevant region of $2$-dominance is where $k_0 < k < k_2$ and it is further divided into three regions as a result of fixed point analysis. Moving horizontally, as $\beta$ increases from $0$ to $1$, the closed loop stable equilibrium loses stability and the system goes into steady oscillations. Eventually, oscillations disappear in favor of a bistable behavior (large positive feedback). In comparison, moving vertically, the effect of increasing $k$ leads to a loss of stability followed by either oscillations or multi-stability, controlled by the value of $\beta$.    
	
	The comparison {between strong and weak time-scale separation cases shows how the region of oscillation shrinks in the weak case. In other words, reducing the separation of time scales makes oscillations less robust. When positive and negative feedback lags have smaller time scale difference, the stabilizing action of the negative feedback is more effective. Hence, achieving oscillations require a stronger positive feedback to occur (larger $\beta$)}.
	
	The closed-loop behavior {associated to different regions in Figure \ref{fig:dominance_map} is illustrated by the  simulations of Figure \ref{fig:simu}}. 
	\begin{figure}[htbp]
		\centering
		\includegraphics[width=0.95\columnwidth]{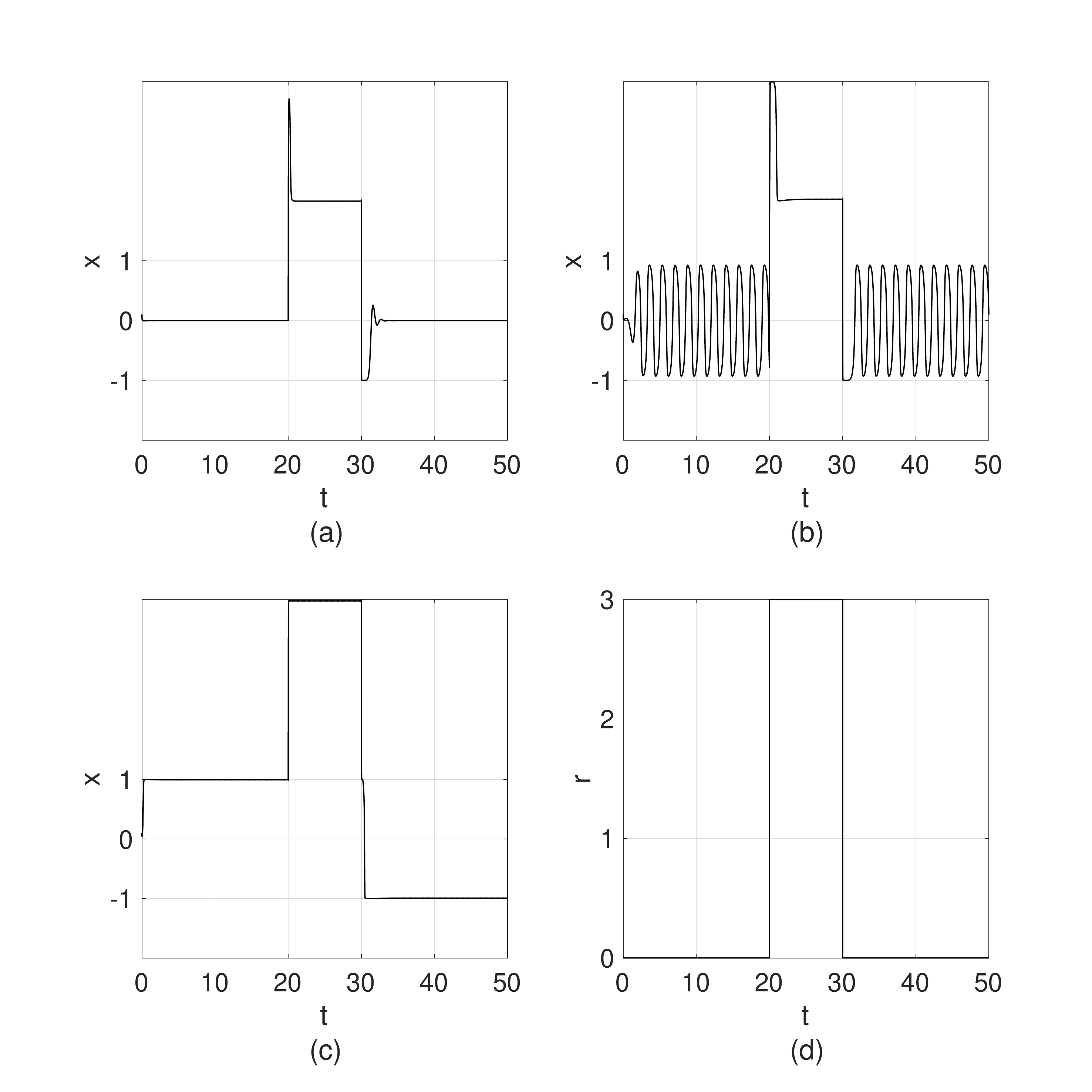} \vspace*{-3mm}
		\caption{Trajectories of the {output $x$ of the mixed feedback closed loop of Figure \ref{fig:dominance_map}.a
		  in response to the reference input $r$ represented in \textbf{(d)}. Initial condition $x=0.1$,  $x_p=0$ and $x_n=0$.}
		  \textbf{(a)}: $k=5$, $\beta=0.2$ ($0$-dominant); \textbf{(b)}: $k=5$, $\beta=0.4$ ($2$-dominant with oscillation); \textbf{(c)}: $k=5$, $\beta=0.8$ ($2$-dominant with both oscillation and equilibria).}
		\label{fig:simu}
	\end{figure}
		
	{In relation with the regions outlined in Figure \ref{fig:dominance_map}.a, 
	the simulations in Figure \ref{fig:simu} refer to control parameters 
	$(k,\beta)$} in the $0$-dominance region (Figure \ref{fig:simu}.a), in the $2$-dominance with oscillation region (Figure \ref{fig:simu}.b), and in the $2$-dominance with oscillation+fixed points region (Figure \ref{fig:simu}.c). {The system is driven by the reference signal $r$ in Figure \ref{fig:simu}.d.}
The state trajectories of the closed loop system in $0$-dominant region and $2$-dominant with oscillation region converge back to the original steady state for $t \geq 30$, as shown in Figure \ref{fig:simu}.a and b. {The bistability of the system in the region of co-existence of 
oscillations and multiple fixed points is well captured by Figure \ref{fig:simu}.c, which shows that the output does not return to the initial steady state for $t\geq 30$.}

	\begin{remark}
		The analysis for $r \neq 0$ is similar. In general, a constant non-zero reference input $r$ will reduce the parameter range for oscillations, since an increase in $|r|$ will stabilize the unstable equilibria by shifting them outside the gray regions in Figure \ref{fig:fixed_point_stability}. $\hfill\lrcorner$
	\end{remark}
	\begin{remark}
		The design of $k$ and $\beta$ for oscillations can be combined with the classical describing function and fast-slow approximation methods to regulate the oscillation frequency, as shown in \cite{che2021shaping}.  	 $\hfill\lrcorner$
	\end{remark}
	
	{The combination of circle criterion for dominance (Theorem \ref{th:circle_cirteria}) 
	and of Nyquist criterion for instability of closed-loop equilibria
	 opens the way to robust analysis (and design). Mimicking classical stability theory, the robustness of the closed-loop oscillations 
	 to plant uncertainties is captured by the maximal perturbation that the Nyquist locus can undertake before 
	 (i)~entering the shaded region in Figure \ref{fig:circle_criteria_zero} (robustness measure for $2$-dominance) and
	 (ii)~changing the number of turns around the $-1$ point of the complex plane, for each equilibrium 
	 (robustness of the instability of each equilibrium). This leads to 
	 quantifiable bounds on the multiplicative/additive plant uncertainties that the closed loop can sustain while preserving 
	 oscillations.} 
	
	Consider a \emph{bounded, fast} (poles lie to the left of the dominance rate $-\lambda$) additive uncertainty, $\mathcal{P}_\Delta(s)=\mathcal{P}(s)+\Delta(s)$. Then, the perturbed transfer function reads
	{\begin{equation}
		G_\Delta(s,k,\beta)=G(s,k,\beta)+\Delta(s)\mathcal{C}(s,k,\beta).
	\end{equation}
	To apply the circle criterion for dominance  we consider the shifted transfer function $G_\Delta(s-\lambda,k,\beta)$. The shifted perturbation on the nominal plant reads $\Delta(s-\lambda)\mathcal{C}(s-\lambda,k,\beta)$, which corresponds to a graphical perturbation on the
	nominal Nyquist locus bounded by 
	$$\delta := \sup\nolimits\limits_{\omega \in[0, \infty)} |\Delta(j\omega-\lambda)\mathcal{C}(j\omega-\lambda,k,\beta)|.$$
	Thus, by continuity, for $\delta$ sufficiently small (i.e. for $|\Delta(j\omega-\lambda)|$ sufficiently small at each frequency $\omega$), the Nyquist plot of $G_\Delta(s-\lambda,k,\beta)$ must remain to the right of the vertical axis passing through $-1$. That is, the perturbed closed-loop system remains $2$-dominant. A similar argument applies to the stability/instability of each equilibrium.}
	
	As an illustration, consider the nominal mixed feedback closed loop system in Figure \ref{fig:dominance_map}.a, with $k=5$ and $\beta=0.4$. Figure \ref{fig:circle_criteria_robustness} shows that the inflated Nyquist plot of {$G(s-\lambda,k,\beta)$} remains to the left of $-1$ for $\delta = 0.95$. The $2$-dominance of the closed loop is thus robust to perturbations $\Delta$ whose poles have real part smaller than $-\lambda$ and that satisfy {$|\Delta(j \omega-\lambda)| \leq \frac{\delta }{|\mathcal{C}(j\omega-\lambda,k,\beta))|}$} for all $\omega \in \mathbb{R}$ (Figure \ref{fig:circle_criteria_robustness} Right).
	
	\begin{figure}[!h]
		\centering
		\includegraphics[width=0.5\textwidth]{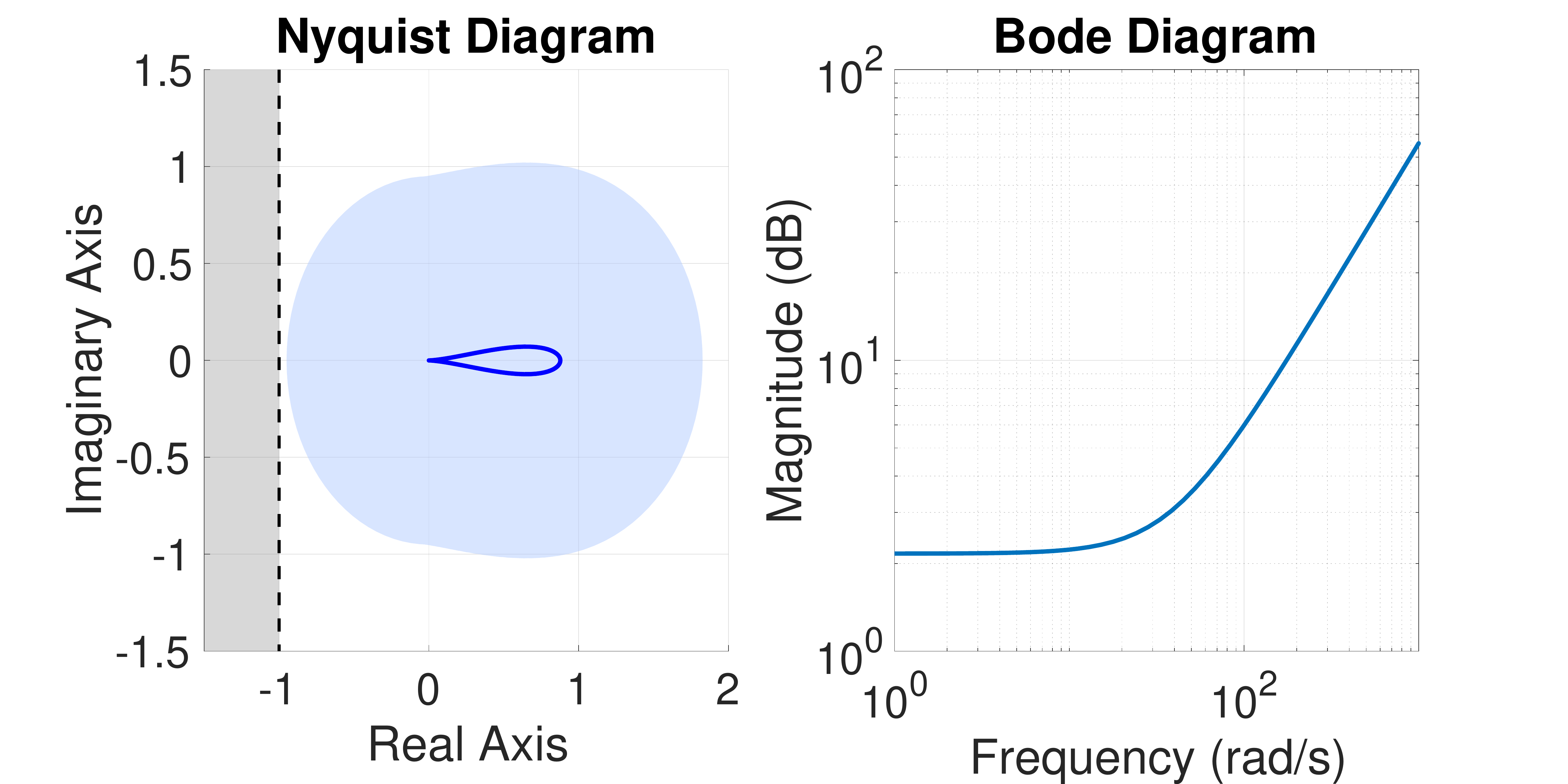}
		\caption{\textbf{Left}: Nyquist plot of $G(s-\lambda)$ for $k = 5$ and $\beta = 0.4$. Recall that $\lambda =50$. The uncertainty bound is $\delta = 0.95$. \textbf{Right}: The upper bound of $|\Delta(j\omega-\lambda)|$ set by {$\frac{\delta }{|\mathcal{C}(j\omega-\lambda,k,\beta))|}$}.}
		\label{fig:circle_criteria_robustness}
	\end{figure}
	
	For the robustness of the oscillations, a further check on the instability of the equilibria of the closed system is needed. The equilibria are now given by \eqref{eq:fixed_point1} computed for $G_\Delta$ instead of $G$. This means that \eqref{eq:fixed_point2} becomes
	\[\frac{y_1}{(k\mathcal{P}(0)+\Delta(0))(2\beta-1)}+r=\varphi(y_1).\] 
	This shows that the perturbation on the position of the equilibria is a function of the DC gain $\Delta(0)$. Their stability/instability can be studied via local analysis, using the Nyquist criterion on the system linearization at each equilibrium. This amount to classical robustness analysis (for stability/instability) and enforces additional constraints on the perturbation bound $\delta$. 
	
	{The analysis of plant perturbations in the example shows} that mixed feedback controller guarantees (quantifiable) robustness of closed-loop oscillations. From a system-theoretic perspective, our analysis supports related observations from system biology and neuroscience, {which identify in the interplay between fast positive feedback and slow negative feedback a source of robustness for oscillators}. 
	
	{
	\begin{remark}
	We have focused on \emph{fast dynamic uncertainties} for reasons of simplicity. These guarantee that the poles of the
	perturbed plant remain to the left of the $-\lambda$ axis. In the context of classical robust stability, this assumption 
	would correspond to the restriction to stable uncertainties. The analysis of robustness to a wider class of uncertainties
	requires a more general version of the the circle criterion for dominance \cite{miranda2018analysis}, and an
	extended use of the Nyquist criterion for the robustness analysis of stable/unstable equilibria. $\hfill\lrcorner$
	\end{remark}
	}
	
	The simplicity of the analysis above, mostly due to the adaptation of classical tools for linear robust stability to robust oscillations, leaves open questions about design. Finding gain $k$ and balance $\beta$ to achieve a prescribed level of robustness remains a challenging problem. Optimizing such parameters using Nyquist analysis is hard and typically requires several iterations. This motivate the development of a more systematic and general framework to support design, based on $p$-dissipativity and LMIs.
		
	\section{$p$-dissipativity and LMIs}
	\label{Section:p-dissipativity}
	\subsection{$p$-dissipativity}
	{As in Section \ref{Sec:Dominance_Th}, we summarize here a minimal set of results 
	on differential dissipativity for open dominant systems. More details can be found in 
	\cite{forni2018differential,miranda2018analysis}. These results will be used 
	to develop control design for robust mixed-feedback oscillators
	based on LMIs. } 
	
	Consider the open nonlinear system of the form:
	\begin{equation}\label{theorem:dissipativity_sys}
		\begin{cases}
			\dot{x}=f(x)+Bu\\
			y=Cx+Du
		\end{cases}\quad x\in\mathbb{R}^n, \ y,u\in\mathbb{R}^m
	\end{equation}
	{where $B \in \mathbb{R}^{n \times m}$, $C \in \mathbb{R}^{m \times n}$, and $D\in \mathbb{R}^{m \times m}$}.
	The prolonged system, derived through linearization, reads
	\begin{equation}\label{theorem:prolonged_dissipative_sys}
		\begin{cases}
			\dot{x}=f(x)+Bu\\
			\delta\dot{x}=\partial f(x)\delta x+B\delta u\\
			y=Cx+Du\\
			\delta y=C\delta x+D\delta u 
		\end{cases}
		{\begin{array}{c}
		(x,\delta x)\in\mathbb{R}^{2n} \, , \\ 
		(u,\delta u),(y,\delta y)\in\mathbb{R}^{2m} \ .
		\end{array}}
	\end{equation}

	\begin{definition}{\cite[Definition 3]{forni2018differential}}\label{de:p-dissipativity}
		The nonlinear system \eqref{theorem:dissipativity_sys} is \emph{differentially $p$-dissipative with rate} $\lambda\geq0$ and differential supply rate:
		\begin{equation}
			\label{eq:s}
			\begin{bmatrix}
				\delta y\\\delta u
			\end{bmatrix}^T\begin{bmatrix}
				Q&L\\L& R
			\end{bmatrix}\begin{bmatrix}
				\delta y\\\delta u
			\end{bmatrix}
		\end{equation} if there exist some symmetric matrix $P$ with inertia $(p,0,n-p)$ and some constant {$\varepsilon\geq0$}, such that the prolonged system \eqref{theorem:prolonged_dissipative_sys} satisfies the conic constraint
		\begin{equation} \label{eq:dissipativity_LMI}
			\begin{bmatrix}
				\delta\dot{x}\\
				\delta x
			\end{bmatrix}^T \begin{bmatrix}
				0&P\\P&2\lambda P+\varepsilon I
			\end{bmatrix}\begin{bmatrix}
				\delta\dot{x}\\
				\delta x
			\end{bmatrix}\leq \begin{bmatrix}
				\delta y\\\delta u
			\end{bmatrix}^T\begin{bmatrix}
				Q&L\\L& R
			\end{bmatrix}\begin{bmatrix}
				\delta y\\\delta u
			\end{bmatrix}
		\end{equation}
		for all $(x,\delta x) \in \mathbb{R}^{2n}$ {and all $(u,\delta u)\in \mathbb{R}^{2m}$. 
		The property is strict if $\varepsilon>0$.}
		$\hfill\lrcorner$
	\end{definition}
	
	\eqref{eq:dissipativity_LMI} {corresponds} to a dissipation inequality 
	of the form $\dot{V}(\delta x) \leq s(\delta y,\delta u)$ for
	$V(\delta x) = \delta x^T P \delta x$ and for $s(\delta y,\delta u)$ given by \eqref{eq:s}, applied to a prolonged system with shifted Jacobian $\partial f(x) + \lambda I$. $p$-dissipativity replaces the usual constraint on the positivity of the storage, i.e. $P > 0$, with a constraint on its inertia.
	
	{In this paper we will make use of two particular supplies. For robustness purposes, we will rely on $p$-dissipativity with respect to} 	the \emph{gain supply}
	\begin{equation} \label{p_gain_LMI}
		\begin{bmatrix}
			\delta y\\\delta u
		\end{bmatrix}^T\begin{bmatrix}
			-I &0\\0 & \gamma^2I
		\end{bmatrix}\begin{bmatrix}
			\delta y\\\delta u
		\end{bmatrix}
	\end{equation}
	where $\gamma$ characterizes the $p$-\emph{gain} of the system. {We will also take advantage of the \emph{passivity supply} for interconnection purposes.}
	\begin{equation}\label{eq:excess/shortage_rate}
		\begin{bmatrix}
			\delta y\\\delta u
		\end{bmatrix}^T\begin{bmatrix}
			-\alpha I & I\\
			I&\mu I
		\end{bmatrix}\begin{bmatrix}
			\delta y\\\delta u
		\end{bmatrix}
	\end{equation}
	where $\alpha>0$ ($\mu<0$) denotes excess of output (input) passivity, and $\alpha<0$ ($\mu>0$) denotes shortage of output (input) $p$-passivity, respectively.
	
	The notion of $p$-gain {combined with the the small gain theorem for dominance \cite{padoan2019h,forni2018differential} provide} a framework for robust control of dominant systems, as in classical robust stability theory.
	\begin{theorem}[Small gain interconnection] \label{th:differential_small_gain}
		Let $\Sigma_i$ be a $p_i$-dominant system with input $u_i$, output $y_i$, and differential $p$-gain $\gamma_i\in\mathbb{R}_+$ with rate $\lambda>0$, for $i\in\{1,2\}$. Then the closed system $\Sigma$ defined by the feedback interconnection
		\[u_1=y_1,\quad u_2=y_2\]
		of $\Sigma_1$ and $\Sigma_2$ is $(p_1+p_2)$-dominant with rate $\lambda$ if $\gamma_1\gamma_2<1$.$\hfill\lrcorner$
	\end{theorem}
	
	Like classical passivity, $p$-passivity {is preserved by negative feedback}, as clarified by the next theorem (\cite[Theorem 4]{forni2018differential}).
	
	\begin{theorem}\label{th:Interconection_theorem_general_passive_Ver}
		Let $\Sigma_i$ be a $p_i$-passive {system} from input $u_i$ to output $y_i$ with dominant rate $\lambda\geq0$ and supply rate
		\begin{equation}\label{eq:general_passivity_LMI}
			\begin{bmatrix}
				-\alpha_i I & I\\
				I&\mu_i I
			\end{bmatrix}\begin{bmatrix}
				\delta y_i\\\delta u_i
			\end{bmatrix} 
		\end{equation}
		for $i\in\{1,2\}$.
		Then the closed loop system defined by the negative feedback interconnection
		\[u_1=-y_2,\quad u_2=y_1\]
		is $p_1+p_2$-dominant if $\alpha_1-\mu_2\geq0$ and $\alpha_2-\mu_1\geq0$. $\hfill\lrcorner$
	\end{theorem} 
	
	Theorem \ref{th:Interconection_theorem_general_passive_Ver} suggests that the shortage of input (output) passivity of one subsystem can be compensated by the excess of output (input) passivity of the other system. Furthermore, as in classical passivity, for $\alpha_i = 0$ and $\mu_i = 0$, the closed loop given by $u_1 = - y_2 + v_1$ and $u_2 = y_1 + v_2$ is also $(p_1+p_2)$-passive from $v = (v_1,v_2)$ to $y = (y_1,y_2)$.
	
	\subsection{Convex relaxation for LMI design}
	\label{sec:convex_relaxation_LMI}
	Conditions \eqref{eq:dominance_LMI} and \eqref{eq:dissipativity_LMI} result in a family of infinite LMIs. Their solutions can be obtained via convex relaxation, {by confining the system Jacobian within the convex hull of a finite set of linear matrices 
	$\mathcal{A}:= \{A_1, \dots , A_N\}$ for $\partial f(x)$, \cite{lmibook}. Namely, for all $x$, we need that} $\partial f(x) = \sum_{i=1}^N \rho_i(x) A_i$ for some $\rho_i(x)$ satisfying $\sum_i \rho_i(x) = 1$ \cite[Section VI.B]{forni2018differential}. 
	
	{Recall that} condition \eqref{eq:dominance_LMI} is equivalent to
	\begin{equation}\label{eq:convex_relaxation}
		\partial f(x)^TP+P\partial f(x)+2\lambda P+\varepsilon I \leq 0, \quad \forall x\in\mathbb{R}^n.
	\end{equation} 
	{Thus, if} $\partial f(x)\in \text{ConvexHull}(\mathcal{A})$ for all $x$, then any (uniform) solution $P$ to
	\begin{equation}\label{eq:finite_set_matrices}
		A_i^TP+PA_i+2\lambda P+\varepsilon I \leq 0,\quad i\in\{1,...,N\}
	\end{equation}
	is also a solution to LMI \eqref{eq:convex_relaxation}. 
	
	Likewise, for the supply rate \eqref{p_gain_LMI} and \eqref{eq:excess/shortage_rate}, solutions to \eqref{eq:dissipativity_LMI} can be obtained by finding a solution $P$ to 
	\begin{equation}
		\label{eq:p_gain_LMI_design}
		\begin{bmatrix}
			A_i^TP+PA_i+2\lambda P +\varepsilon I & PB &C^T\\
			B^TP & -\gamma I & D^T\\
			C & D & -\gamma I
		\end{bmatrix}\leq 0
	\end{equation}
	and 
	\begin{equation} \label{eq:LMI_excess/shortage_passivity}
		\begin{bmatrix}
			A_i^TP+PA_i+2\lambda P +\varepsilon I & PB-C^T&C^T\\
			B^TP-C & -\mu I & D^T\\
			C & D & -\frac{1}{\alpha}I
		\end{bmatrix}\leq 0
	\end{equation} respectively, {where $i\in\{1,...,N\}$.}
		
	These inequalities correspond to classical gain and passivity inequalities \cite{lmibook}. The main difference for dominant systems is that $P$ is not necessarily positive definite but satisfies a constraint on its inertia. {This constraint cannot be enforced explicitly while remaining within the language of LMIs. Fortunately,} if the rate lambda splits the eigenvalues of each matrix $A_i$ into a group of $p$ eigenvalues to the right of $-\lambda$ and $n-p$ eigenvalues to left of $-\lambda$, then any solution $P$ will have inertia $(p,0,n-p)$.
	{
	\begin{remark}
		Finding a tight convex hull relaxation is not a trivial task, in general.
	However, for the mixed feedback controller, the presence of a single 
	nonlinearity $\varphi$ with bounded slope $0 \leq \partial \varphi \leq 1$ 
	strongly simplifies the problem. A tight convex hull
	is given by a family of two matrices, corresponding
	to the system's Jacobian associated to $\partial \varphi = 0$ and $\partial \varphi = 1$. $\hfill\lrcorner$
	\end{remark}	
	}
	
	\section{{LMI design} for closed-loop oscillations}
	\label{Sec:State_Feedback_Design}
	\subsection{State feedback design}	
	In this section we adapt the LMIs \eqref{eq:finite_set_matrices}, \eqref{eq:p_gain_LMI_design}, and \eqref{eq:LMI_excess/shortage_passivity} for control purposes, with the goal of finding a state-feedback that guarantees oscillations.
	We first focus on state-feedback design for $2$-dominance, to guarantee a landscape of simple nonlinear attractors. Then, we provide additional conditions to ``destabilize" the system equilibrium (at least one), to enforce multistability and oscillations.
	
	We consider the {generalized mixed feedback closed loop represented} in Figure \ref{fig:state_feedback_block}. As in Section \ref{Section:Root_Locus}, we assume that the dynamics of the plant $\mathcal{P}$ is faster than the dynamics of the mixed feedback controller. We replace gain and balance parameters of Sections \ref{Section:Root_Locus} and \ref{Section:Sufficiency} with generic feedback gains on the components of the mixed feedback controller $x_p$ and $x_n$, and we extend the control action with a full state feedback from the plant state $x_0\in \mathbb{R}^n$. The system has the state space representation:
	\begin{equation} \label{eq:LMI_SS_general}
		\begin{cases}
			\dot{x}=Ax+Bu\\
			y=Cx\\
			u=\varphi(Kx)
		\end{cases} 
		\quad x=\begin{bmatrix} x_0 \\  x_p \\ x_n \end{bmatrix} \in\mathbb{R}^{n+2}.
	\end{equation}
	\begin{figure}[htbp]
		\centering
		\includegraphics[width=0.3\textwidth]{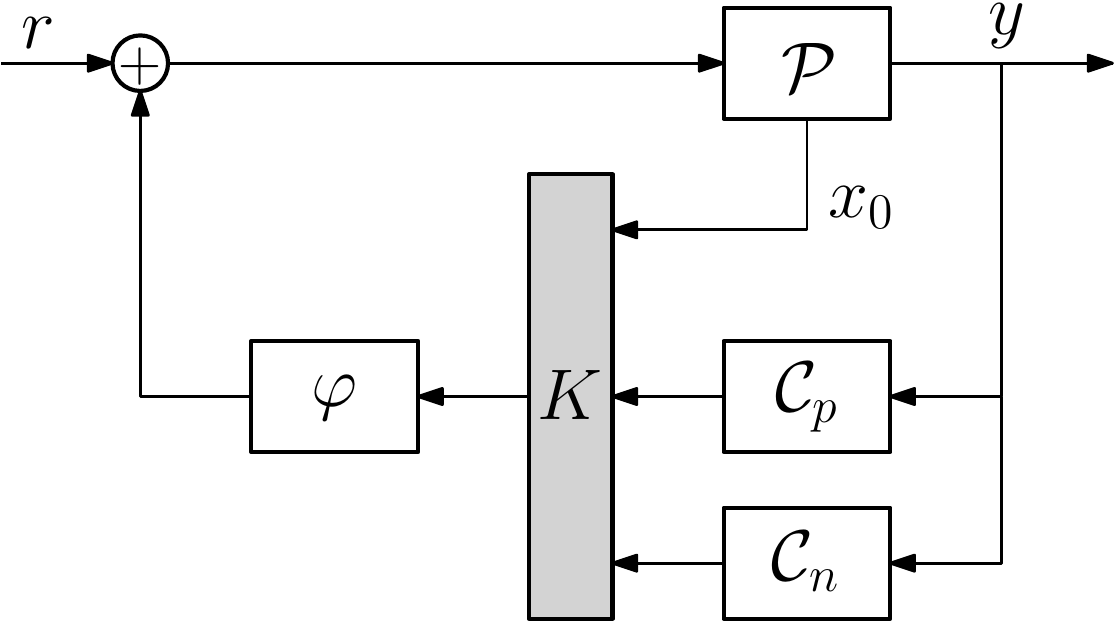}
		\caption{Block diagram of the {generalized mixed feedback closed loop}.}
		\label{fig:state_feedback_block}
	\end{figure}
	
	\begin{theorem} \label{th:state_feedback_design}
		The state-feedback matrix $K$ guarantees $2$-dominance in closed loop with rate $\lambda$ if there exist a symmetric matrix $Y$ with inertia $(2,0,n)$, a matrix $Z$, and $\varepsilon > 0$ such that
		\begin{equation}\label{eq:LMI_2_dominant}
			\begin{cases}
				YA^T + AY+2\lambda Y + \varepsilon I\leq0\\
				YA^T + Z^TB^T + AY + BZ +2\lambda Y +\varepsilon I\leq0.
			\end{cases}	
		\end{equation}
		The state feedback gain $K$ is given by $K=ZY^{-1}$.$\hfill\lrcorner$
	\end{theorem}
	\begin{proof}
		With the state feedback gain $K$, the prolonged system of \eqref{eq:LMI_SS_general} is
		\begin{equation*}
			\begin{cases}
				\dot{x}=Ax+B\varphi(Kx)\\
				\delta\dot{x}=(A+B\partial\varphi(Kx)K)\delta x
			\end{cases}
		\end{equation*}
		Since $\partial\varphi\in[0,1]$, the set $\mathcal{A}:=\{A,A+BK\}$ guarantees $(A+B\partial\varphi(Kx)K)\in $ \textit{ConvexHull}$(\mathcal{A})$ for all $x$. By convex relaxation, the system \eqref{eq:LMI_SS_general} is $2$-dominant if there exist a matrix $K$, a symmetric matrix $P$ with inertia $(2,0,n)$, and $\epsilon>0$ such that:
		\begin{equation}\label{eq:theorem7_LMI}
			\begin{cases}
				A^TP + PA+2\lambda P+\epsilon I\leq0\\
				(A+BK)^TP + P(A+BK) +2\lambda P+\epsilon I\leq0\\
			\end{cases}	
		\end{equation}
		Let $Y=P^{-1}$, $Z=KY$ and $\varepsilon = \epsilon Y Y$. Then, by pre- and post-multiplying \eqref{eq:theorem7_LMI} by $Y$, we obtain \eqref{eq:LMI_2_dominant}.
	\end{proof}
	
	Feasibility of \eqref{eq:LMI_2_dominant} follows from Section \ref{Section:Sufficiency}, since the selection of gain and balance corresponds to a particular state feedback $K$ {(specifically, feasibility follows from Theorem
	\ref{th:2-dominance}, as a consequence of Theorem \ref{th:circle_cirteria}).}
	The inertia constraint on $Y$ (as well as on $P$) makes the optimization problem non-convex. However, as in Section \ref{Section:p-dissipativity}.B, there is no need to enforce this constraint explicitly. The first inequality in \eqref{eq:LMI_2_dominant} guarantees that $Y$ has inertia $(2,0,n)$ whenever two eigenvalues of $A$ fall to the right of $-\lambda$. This also implies that the plant dynamics limit the design of the closed loop, {by enforcing a constraint on the time constants $\tau_p$ and $\tau_n$, which} must be sufficiently slow. {This affects the achievable oscillation frequency}. 
	
	\begin{remark}
	\label{rem:precompensator}
	{To recover design flexibility, a pre-compensation feedback could be introduced with the goal of making the open loop dynamics faster. This is illustrated in Figure \ref{fig:Prestabilize_Feedback}. Consider the plant dynamics $\mathcal{P}_0$}
	 \begin{equation}\label{eq:P_0_SS}
		\begin{cases}
			\dot{x}_0=A_0x_0+B_0u_0\\
			y=C_0x_0
		\end{cases}
	\end{equation}
	For any given rate $\lambda$, the pre-stabilizing state-feedback matrix $K_0$ must guarantee that all the poles of $\mathcal{P}$ lies to the left of $-\lambda$. Under controllability assumptions of $(A_0,B_0)$, this is guaranteed by the additional LMI condition
	\begin{equation}
		Y_0A_0^T + Z_0^TB_0^T + A_0Y_0 + B_0Z_0 +2\lambda Y_0 + \varepsilon I\leq0 
	\end{equation}
	in the unknowns $Y_0 \!=\! Y_0^T \!>\! 0$ and $Z_0$. Thus, $K_0\!=\!Z_0Y_0^{-1}$. $\hfill\lrcorner$
	\end{remark}
	
	\begin{figure}[htbp]
		\centering
		\includegraphics[width=0.25\textwidth]{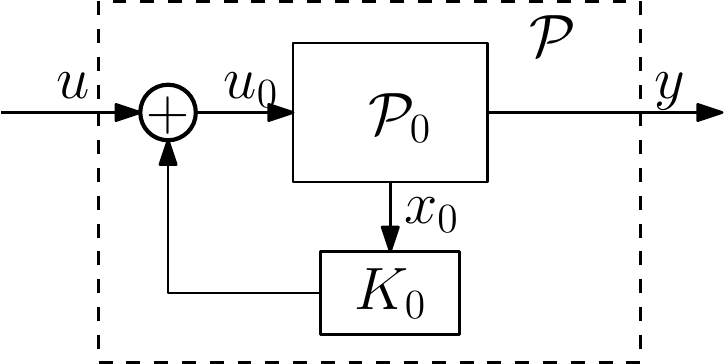}
		\caption{Pre-conditioning feedback.}
		\label{fig:Prestabilize_Feedback}
	\end{figure} 
	
	To induce a stable oscillation in closed loop, we {combine \eqref{eq:LMI_2_dominant} with} the following constraint
	\begin{equation}\label{eq:LMI_unstable_origin}
		YA^T + Z^TB^T + AY + BZ +\varepsilon I\leq 0.
	\end{equation}
	For $r = 0$ ($r \neq 0$ is similar), the constraint on the inertia of $Y$ in Theorem \ref{th:state_feedback_design} combined with \eqref{eq:LMI_unstable_origin} guarantee that the equilibrium point at the origin is unstable. 
	
	In agreement with Section \ref{Section:Sufficiency}, the DC gain of the linear open loop component is $-KA^{-1}B$, that is, the slope of the line in Figure \ref{fig:fixed_point_stability} is now $\frac{1}{-KA^{-1}B}$. This implies that the system will oscillate for 
	``low" gains $K$ and will either oscillate or show multiple equilibria for ``high" gains $K$. Specifically, 
	the closed loop has a single equilibrium if $-KA^{-1}B < 1$, which is unstable by \eqref{eq:LMI_unstable_origin}. This guarantees stable oscillations in closed loop (given the boundedness of the closed-loop trajectories). Multiple equilibria will appear for $-KA^{-1}B>1$, which may lead to a region of co-existence of oscillations and stable fixed points.
	
	To reduce the control gains $|K|$ when $-KA^{-1}B > 1$, the following constraint can be added
	\begin{equation}
		\begin{bmatrix}
			-\nu & Z\\ Z^T& -I
		\end{bmatrix}\leq 0 \ ,
	\end{equation}
	{where the constant $\nu>0$} limits the norm square of matrix $Z$, i.e. by Schur complement $ZZ^T\leq \nu$. 
	Since $K = Z{Y^{-1}}$, if $Y$ does not change dramatically, the parameter $\nu$ effectively control the magnitude of $K$.
	
	\begin{remark}
	{The combination of \eqref{eq:LMI_2_dominant} and \eqref{eq:LMI_unstable_origin}
	leads to the automatic derivation of state-feedback gains for oscillations 
	The design takes advantage of the particular feedback structure in Figure \ref{fig:state_feedback_block},
	which is limited to linear plants for simplicity.
	The approach can be extended to nonlinear plants of the form
	$\dot{x} = f(x) + B u$ by taking advantage of more general convex hull relaxations, as discussed 
	in Section \ref{sec:convex_relaxation_LMI}.} $\hfill\lrcorner$
	\end{remark}
		
	\subsection{Example: mixed state-feedback of a first order plant}
	\label{sec:example2_LMIfirstorder}
	For illustration, we revisit the design of Section \ref{Section:Example1} using LMIs. The linear component has matrices
	\begin{equation}\label{eq:LMI_SS_example1}
		\begin{split}
			&A=\begin{bmatrix}
				-\frac{1}{\tau_l}&0&0\\
				\frac{1}{\tau_p}&-\frac{1}{\tau_p}&0\\
				\frac{1}{\tau_n}&0&-\frac{1}{\tau_n}
			\end{bmatrix},\quad
			B=\begin{bmatrix}
				\frac{1}{\tau_l}\\0\\0
			\end{bmatrix}\\
			&C=\begin{bmatrix}
				1& 0& 0
			\end{bmatrix}
		\end{split}
	\end{equation}
	where $\tau_l=0.01$, $\tau_p=0.1$, $\tau_n=1$. Setting $\lambda=50$ and {solving}
	\eqref{eq:LMI_2_dominant} and \eqref{eq:LMI_unstable_origin} with CVX \cite{cvx}, we get
	\[Y=\begin{bmatrix}
		0.3788&   -0.8923&   -0.2650\\
		-0.8923&   -0.5368&   -0.2545\\
		-0.2650&   -0.2545&   -0.2053
	\end{bmatrix}.\] 
	$Y$ has inertia $(2,0,1)$ and the controller gains read
	\[K=ZY^{-1}=\begin{bmatrix}
		0.5284 &   0.9623  & -0.6342
	\end{bmatrix}.\]
	The DC gain $-KA^{-1}B=0.8565<1$ guarantees a unique unstable equilibrium point and hence stable oscillations for $r=0$ (Figure \ref{fig:LMI_dominance_eg}, left).
	
	The LMI design approach can also be leveraged to handle parametric uncertainties, for example on the time constants of the mixed feedback controller $\mathcal{C}(s,k,\beta)$, as it is often the case in the biological setting. 	
	
	Suppose that $\tau_p$ and $\tau_n$ in \eqref{eq:LMI_SS_example1} are affected by a $20\%$ perturbation, i.e. $\tau_p\in[0.08, 0.12]$, $\tau_n\in[0.8,1.2]$. This variability can be taken into account by extending \eqref{eq:LMI_2_dominant} and \eqref{eq:LMI_unstable_origin} to the convex-hull of matrices given by the four combinations $(\tau_p,\tau_n)\in\{(0.08,0.8),(0.08,1.2),(0.12,0.8),(0.12,1.2)\}$. The solution returned by the CVX is
	\[Y=\begin{bmatrix}
		1.9803&   -3.4094&   -1.1150\\
		-3.4094&   -2.2296&   -1.1118\\
		-1.1150&   -1.1118&   -0.6575
	\end{bmatrix}\]
	which has inertia $(2,0,1)$. The controller gains are
	\[K=ZY^{-1}=\begin{bmatrix}
		0.5296 &    1.3804 &  -1.6173
	\end{bmatrix}.\]
	For nominal values, the DC gain $-KA^{-1}B=0.2926<1$. The {DC gain remains smaller than $1$ also for perturbed time constants, which} guarantees a unique unstable equilibrium point, that is, stable oscillations for $r=0$ (Figure \ref{fig:LMI_dominance_eg}, right).
	
	\begin{figure}[!h]
		\centering
		\includegraphics[width=0.5\textwidth]{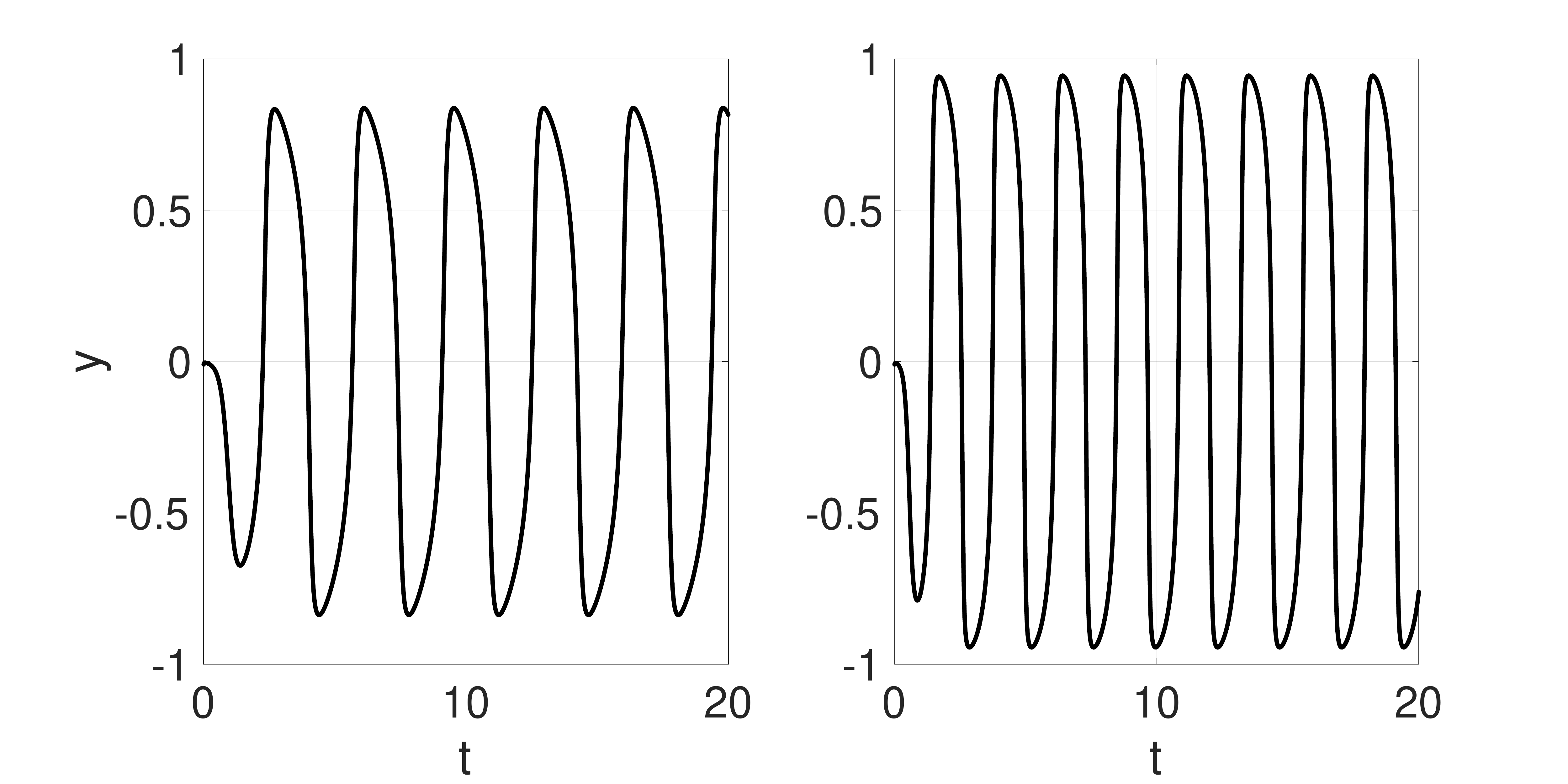}
		\caption{Output $y$ of the closed loop \eqref{eq:LMI_SS_general} for $\tau_l=0.01$, $\tau_p=0.1$, $\tau_n=1$. \textbf{Left}: oscillation for nominal design. \textbf{Right}: oscillation for the robust design.} 
		\label{fig:LMI_dominance_eg}
	\end{figure}
	
	{
	Even if we did not enforce any specific sign pattern on the controller gains, 
	both nominal and perturbed controllers show a negative gain associated with the slow network $\mathcal{C}_n$ and a 
	positive gain associated with the fast network $\mathcal{C}_p$. Remarkably, the gains of the perturbed case have larger magnitude. 
	This is a further confirmation of the role of mixed-feedback in the generation of stable and robust oscillations.} 
		
	\section{{LMI design for} robustness and passivity}
	\label{Sec:Robust_and_Interconnection}
	\subsection{Robust oscillations via robust $2$-dominance}
	{In Section \ref{Section:Example1_robustness} we have characterized the robustness of the oscillations generated via mixed-feedback using graphical arguments based on Nyquist diagrams. Leveraging LMIs, Section \ref{sec:example2_LMIfirstorder} provides a first example of a design procedure that optimizes the controller parameters to guarantee robust oscillations to prescribed parametric uncertainties. In this section we will consider dynamic uncertainties as in classical robust control to develop a design framework for robust oscillations. Our approach 
	takes advantage of small gain results for dominance theory.}
	
	Consider the mixed feedback closed loop in Figure \ref{fig:uncertainty_block}, where dynamic uncertainties are represented by the block $\Delta$. {Note that the uncertain dynamics are not necessarily linear. Furthermore, the figure represents the case of multiplicative uncertainties on the output for simplicity} but our approach is general. The mixed feedback closed loop has the state space representation
	\begin{equation} \label{eq:LMI_Robust_SS_A_B}
		\begin{cases}
			\dot{x}=Ax+B_1u+B_2 w\\
			u=\varphi(Kx)\\
			y=C_1x\\
			z=C_2x			
		\end{cases}
	\end{equation}
	where $A$, $B_1$ and $C_1$ are the nominal state space matrices as \eqref{eq:LMI_SS_general}. $B_2$ and $C_2$ characterize how the uncertain dynamics affect the nominal system. For example, the specific case in Figure \ref{fig:uncertainty_block} takes $B_2=[0\ ...\ 0\ 1/\tau_p\ 1/\tau_n]^T$ and $C_2=C_1$.
	
	\begin{figure}[htbp]
		\centering
		\includegraphics[width=0.38\textwidth]{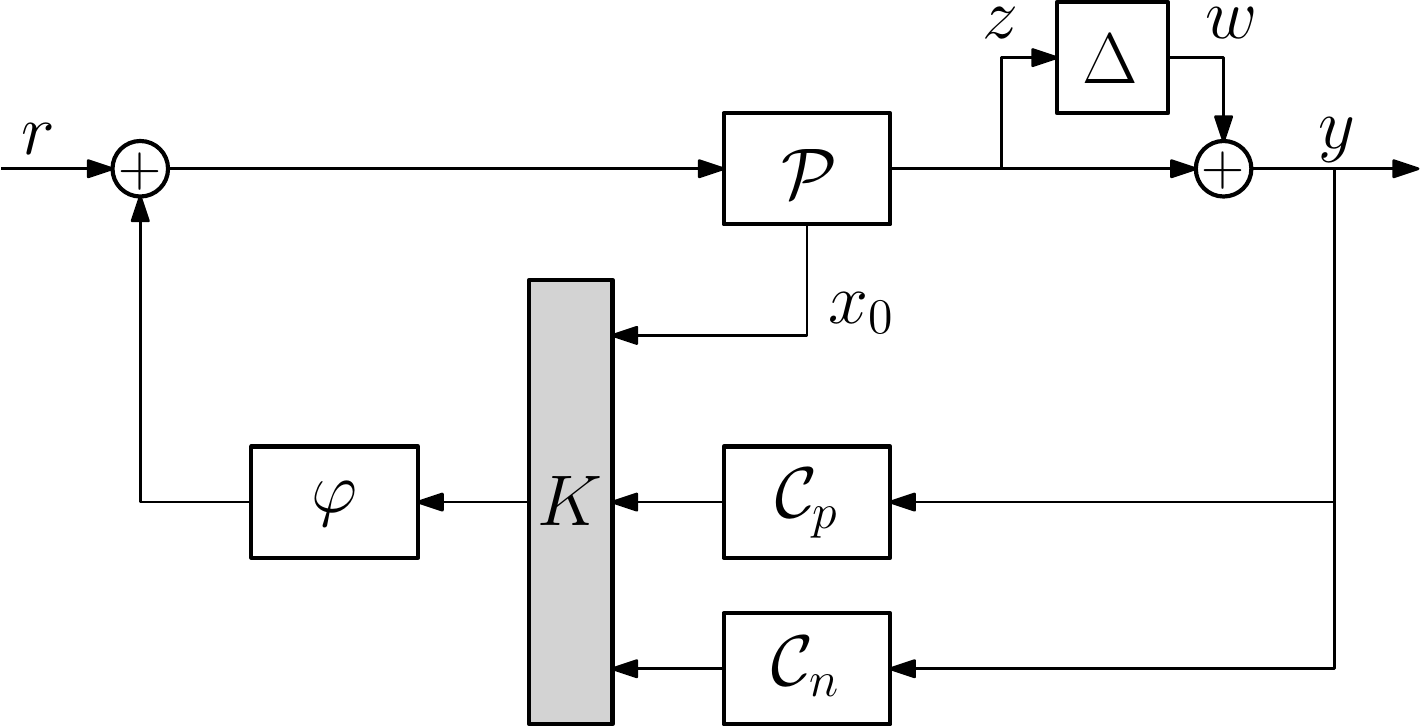}
		\caption{State-feedback for robust oscillations.}
		\label{fig:uncertainty_block}
	\end{figure} 
	
	\begin{theorem}
	\label{th:LMI_Robustness}
		{Suppose that the uncertain dynamics $w=\Delta(z)$ has $0$-gain less than $\frac{1}{\gamma}$ with rate $\lambda$}. Then, the closed loop given by \eqref{eq:LMI_Robust_SS_A_B} and $\Delta$ is $2$-dominant if there exist a symmetric matrix $Y$ with inertia $(2,0,n)$, a matrix $Z$, and $\varepsilon > 0$ such that
		\small
		\begin{subequations}
			\label{eq:Robustness_LMI}
			\begin{align}
				&\begin{bmatrix}\label{eq:Robustness_LMI_1}
					YA^T+AY+2\lambda Y+\varepsilon I & B_2 & YC_2^T\\
					B_2^T & -\gamma I& 0\\
					C_2 Y & 0 & -\gamma  I
				\end{bmatrix} \leq 0\\
				&\begin{bmatrix}
					YA^T\!+\!Z^T\!B_1^T\!+\!AY\!+\!B_1Z\!+\!2\lambda Y\!+\!\varepsilon I & B_2 & YC_2^T\\
					B_2^T & -\gamma I& 0\\
					C_2 Y & 0 & -\gamma  I
				\end{bmatrix} \leq 0 \label{eq:Robustness_LMIb}
			\end{align}	
		\end{subequations}
		\normalsize
		and $K=ZY^{-1}$. $\hfill\lrcorner$
	\end{theorem}
	
	\begin{proof}
		By the differential small gain theorem \ref{th:differential_small_gain}, the $0$-gain $\frac{1}{\gamma}$ of $\Delta$ sets a strict upper bound of the $2$-gain of the mixed feedback closed loop, $\gamma$. Using \eqref{eq:p_gain_LMI_design}, \eqref{eq:LMI_Robust_SS_A_B} has $2$-gain 
		$\gamma$ if there exist a matrix $K$, a symmetric matrix $P$ with inertia $(2,0,n)$, and $\epsilon > 0$ such that
		\begin{equation}
			\label{eq:Robustness_LMI_proof}
			\begin{bmatrix}
				A_i^TP+PA_i+2\lambda P +\epsilon I & PB_2 &C_2^T\\
				B_2^TP & -\gamma I & 0\\
				C_2 & 0 & -\gamma I
			\end{bmatrix}\leq 0
		\end{equation}
		where $A_i\in\{A,A+B_1K\}$. Let $Y=P^{-1}$, $Z=KY$, and $\varepsilon = \epsilon Y Y$, 
		\eqref{eq:Robustness_LMI} is thus obtained by 
		pre- and post-multiplying \eqref{eq:Robustness_LMI_proof} by 
		\[\begin{bmatrix}
			Y&0&0\\
			0&I&0\\
			0&0&I
		\end{bmatrix}. \vspace*{-5mm}
		\] 
	\end{proof}
	
		{
		For large gains $\gamma$, the feasibility of \eqref{eq:Robustness_LMI} reduces to the feasibility of \eqref{eq:LMI_2_dominant}, discussed in Section \ref{Sec:State_Feedback_Design}. The assumption on the dynamic uncertainties guarantees that $\Delta$ belongs to a family of incrementally stable perturbations whose decay rate is faster than $\lambda$. This means that the perturbed plant dynamics remain faster than the controller dynamics, as in the nominal case.}	 
	
	{Together with $2$-dominance, to guarantee robust oscillations we need
	to ensure that the instability of the equilibrium point at the origin is also robust to perturbations. 
	This can be established via robustness criteria for instability (see e.g. \cite{hara2020robust}). 
	In this paper, we take advantage again of small gain results for dominance.}
	\begin{itemize}
		\item 
		We can \emph{verify} that the instability persists. 
		Consider the system linearization at the equilibrium point given by $\bar{A}=A+BK$.
		We can combine \eqref{eq:LMI_2_dominant} and \eqref{eq:LMI_unstable_origin} 
		with 	 \eqref{eq:p_gain_LMI_design} , the latter for $A_i = \bar{A}$, $\lambda=0$, and for $\gamma = \gamma_{\mathrm{ins}}$, to certify that the instability will be preserved by any perturbation $\Delta$ with $0$-gain less than $1/\gamma_{\mathrm{ins}}$ (for $\lambda \!=\! 0$). This follows from Theorem \ref{th:differential_small_gain}.
			\item 
			We can also \emph{design} the feedback $K$ to enforce the desired level of robust instability of the equilibrium.
			To achieve this, we pair \eqref{eq:Robustness_LMI} to an additional LMI of the form \eqref{eq:Robustness_LMIb}
			where we take $\lambda = 0$ and $\gamma = \gamma_{\mathrm{ins}}$.
	\end{itemize}
	\begin{remark}
	 {
	 The control action represented by state-feedback matrix $K$ has no effect
		on the closed loop gains when $\partial\varphi(Kx)=0$. This is formalized by
		\eqref{eq:Robustness_LMI_1}, which shows how the open-loop plant features 
		limit the achievable $2$-gain. Performances can be improved via pre-compensation, 
		following the approach of Section \ref{Sec:State_Feedback_Design}.A.
	 From Remark \ref{rem:precompensator}, 
	 a pre-compensator $K_0$ can be designed to recover performances, as illustratred by Figure \ref{fig:Prestabilize_Feedback}.
	 Using \eqref{eq:P_0_SS} to denote the plant dyamics before pre-compensation, the state feedback $K_0$ 
	 can be computed with} the additional LMI
	\begin{equation}
		\small
		\setlength\arraycolsep{1.5pt}\begin{bmatrix}
			Y_0A_0^T\!+\!Z_0^T\!B_0^T\!+\!A_0Y_0\!+\!B_0Z_0\!+\!2\lambda Y_0\!+\!\varepsilon I&B_0&Y_0C_0^T\\
			B_0^T&-\gamma I& 0\\
			C_0Y_0 & 0 & -\gamma I
		\end{bmatrix} \leq 0
	\end{equation}
	\normalsize
	in the unknowns $Y_0 \!=\! Y_0^T \!>\! 0$, $Z_0$, and $\varepsilon \!> \!0$. $K_0\!=\!Z_0Y_0^{-1}$. $\hfill\lrcorner$
	\end{remark}

	\subsection{$2$-passivity and interconnections}
	The mixed feedback closed loop can also be adapted to passive interconnections,
	taking advantage of Theorem \ref{th:Interconection_theorem_general_passive_Ver}.
	The goal of this section is to design the controller gains
	to achieve $2$-dominance for closed-loop interconnections represented in Figure \ref{fig:Load_block}. 
	We assume that $\mathcal{P}_{ex}$ is a generic external dynamics, $0$-passive 
	with excess of output passivity $\alpha$ at rate $\lambda$. 
	This implies that $\mathcal{P}_{ex}$ has fast transients and its shifted dynamics are incrementally passive.
	\begin{figure}[!h]
		\centering
		\includegraphics[width=0.3\textwidth]{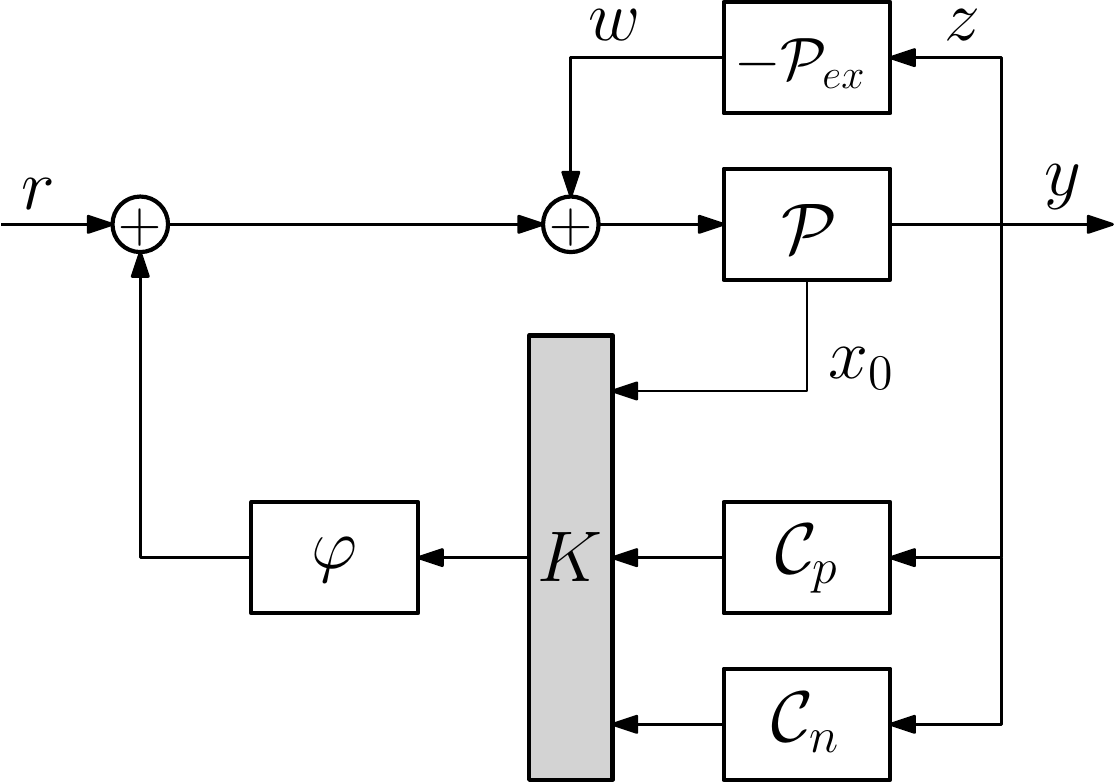}
		\caption{State feedback for passive interconnections.}
		\label{fig:Load_block}
	\end{figure} 
	
	\begin{theorem}
		\label{th:LMI_passive_interconnect}
		Consider a $0$-passive system $\mathcal{P}_{ex}$ with excess of output passivity $\alpha>0$ at rate $\lambda$.
		Then, the closed loop given by \eqref{eq:LMI_Robust_SS_A_B} and $w = -\mathcal{P}_{ex}(z)$ is $2$-dominant 
		if there exist a symmetric matrix $Y$ with inertia $(2,0,n)$, a matrix $Z$, $\mu<\alpha$, and $\varepsilon > 0$ such that
		\begin{subequations}
			\label{eq:Gerenal_passive_interconnection_LMI}
			\begin{align}
				&\begin{bmatrix}
					YA^T\!+\!AY\!+\!2\lambda Y\!+\!\varepsilon I & B\!-\!YC^T\\
					B^T\!-\!CY & -\mu I\\
				\end{bmatrix}\leq0 \label{eq:passivity_first}\\
				&\begin{bmatrix}
					YA^T\!+\!Z^TB^T\!+\!AY\!+\!BZ\!+\!2\lambda Y\!+\!\varepsilon I & B\!-\!YC^T\\
					B^T\!-\!CY & -\mu I
				\end{bmatrix}\leq0
			\end{align}
		\end{subequations}
		and $K=ZY^{-1}$.
	\end{theorem}
	\begin{proof}
		From Theorem \ref{th:Interconection_theorem_general_passive_Ver}, the excess of output $0$-passivity $\alpha$ sets an upper bound of the 
		shortage of input $2$-passivity $\mu<\alpha$ for the mixed feedback closed loop. 
		In this case, the LMI condition \eqref{eq:LMI_excess/shortage_passivity} reads
		\[\begin{bmatrix}
			A_i^TP+PA_i+2\lambda P +\epsilon I& PB-C^T\\
			B^TP-C & -\mu I\\
		\end{bmatrix}\leq 0.\]
		where $A_i\in\{A,A+BK\}$. Set $Y=P^{-1}$,  $Z=KY$, and $\varepsilon = \epsilon \lambda_{min}(YY)$. Then, 
		\eqref{eq:Gerenal_passive_interconnection_LMI} is obtained by pre- and post-multiplying 
		the matrix above by 
		\[\begin{bmatrix}
			Y&0\\
			0&I\\
		\end{bmatrix}. \vspace{-4mm}\]
	\end{proof}
	
	{When the shortage of input passivity is large, $\mu \gg 0$, the feasibility of 
	\eqref{eq:Gerenal_passive_interconnection_LMI} reduces to the feasibility of \eqref{eq:LMI_2_dominant}.
	\eqref{eq:Gerenal_passive_interconnection_LMI} guarantees that the closed loop system given by
	the plant $\mathcal{P}$ and the mixed feedback controller $\mathcal{C}$ is $2$-passive from $w$} {to $z$.
	This property is later used to guarantee that negative feedback interconnections with $0$-passive dynamics preserve $2$-dominance.
	As in the previous section, $2$-dominance is not enough to guarantee 
	robust oscillations. Additional conditions must be enforced to guarantee the instability of
	equilibria}.

%

	\subsection{Example: controlled oscillations in a large circuit modelling (simplified) neural dynamics}
	
	Consider the large circuit in Figure \ref{fig:RCnetwork}. 
	{$\Sigma_a$ is given by the interconnection of the mixed feedback controller 
	with a RC circuit (plant)}. {{$\Sigma_a$ can be considered as a simplified conductance-based model of a neuron, 
	with the mixed feedback controller modelling fast and slow conductances affecting 
	the dynamics of the neuron membrane.}
	\color{black}$\Sigma_b$ represents a spatially discretized cable dynamics, modelling how current and voltage distribute along neurites. 
	Their interconnection satisfies  $v_0^a=v_0^b$ and $i_0^a=-i_0^b$. 
	
	The mixed feedback loop $\Sigma_a$ is given by \eqref{eq:LMI_SS_general}. From  \eqref{eq:LMI_SS_example1},  $A$ is given by $\tau_l=R_0C_0=0.01$ ($R_0=100$ and $C_0=10^{-4}$), and we keep $\tau_p=0.1$ and $\tau_n=1$, as in the other examples. Considering 
	$i_0^a$ as input and $v_0^a$ as output, we have  
	$B=\begin{bmatrix}
		\frac{1}{C_0}&0&0
	\end{bmatrix}^T$ and 
	$C=\begin{bmatrix}
		1&0&0
	\end{bmatrix}$.
	$\Sigma_b$ has input $v_0^b$ and output $i_0^b$. The model is taken from cable theory \cite{tuckwell1988introduction}, where $R_1$ represents the resistance along the fiber and the parallel of $R_2$ and $C_m$ represents the impedance of each segment.  
	For $n$ segments, the admittance of $\Sigma_b$ is recursively described by
	\begin{equation}\label{eq:RC_nseg_TF}
		G_n(s)=\dfrac{1}{R_1+\dfrac{1}{C_ms+1/R_2+G_{n-1}(s)}} , 
	\end{equation}
	with base case
	\begin{equation}
		G_1(s)=\frac{C_mR_2s+1}{C_mR_1R_2s+R_1+R_2}.
	\end{equation}

	
	
	It is easy to show that $G_1(s)$ is positive real (passive). The same result holds for the shifted transfer function $G_1(s-\lambda)$, if $\lambda<1/C_mR_2$, which captures the fact that the zero of $G_1(s)$ lies to the left of $-\lambda$. Under such condition, by induction, $\Sigma_b$ remains $0$-passive for rate $\lambda<1/C_mR_2$, since addition and inversion in \eqref{eq:RC_nseg_TF} preserve passivity, and the elements on the right-hand side of \eqref{eq:RC_nseg_TF} are all positive real. We can further deduce that
	\begin{equation}
		|G_n(j\omega-\lambda)|\leq\frac{1}{R_1},\quad \mbox{if }\lambda<1/C_mR_2,
	\end{equation}
	which indicates that $\Sigma_b$ has an excess of output passivity. As parameters, we take $C_m=C_0=10^{-4}$, $R_1=R_0=100$, and $R_2$ varying in $[300,600]$. 
	
	We consider the mixed state feedback design for passive interconnection and set the dominant rate $\lambda=15$. For all $R_2\in[300, 600]$, $\Sigma_b$ is $0$-passive with an excess of passivity $\alpha>30$, This is verified using \eqref{eq:LMI_excess/shortage_passivity} on a minimal state space realization of \eqref{eq:RC_nseg_TF}. Thus, following Theorem \ref{th:LMI_passive_interconnect}, the state-feedback gains of the mixed feedback loop 
	$\Sigma_a$ are obtained by setting $\mu=30$ in \eqref{eq:Gerenal_passive_interconnection_LMI}. We also enforce \eqref{eq:LMI_unstable_origin} to destabilize the equilibrium at $0$.  The solution
	\[Y=\begin{bmatrix}
		18836.5	&-724.7&	-85.5\\
		-724.7	&-696.9&	-138.6\\
		-85.5	&-138.6&	-81.7
	\end{bmatrix}\]
	has inertia $(2,0,1)$ and the controller gains read
	\[K=ZY^{-1}=\begin{bmatrix}
		-3.1117  &  7.1900 &  -6.5486
	\end{bmatrix}.\]
	The DC gain $-KA^{-1}B=-2.4703<1$ guarantees a unique unstable equilibrium point. The instability of the equilibrium is robust to the interconnection with $\Sigma_b$ since the linearization of $\Sigma_a$ at the origin has $2$-gain  $54.84<1/|G_n(s)|_\infty=100$ (for $\lambda=0$). Thus stable oscillations are guaranteed after interconnection for any length $n$ of the cable $\Sigma_b$, as shown in Figure \ref{fig:RCnetwork_simu} for $n=15$. 
	
	The output $v_0 = v_0^a=v_0^b$  of \eqref{eq:LMI_SS_general} maintains its oscillation pattern for a wide range of $R_2$ values. As the signal travels down the cable $\Sigma_b$ we observe a decay of oscillations magnitude, with the smaller $R_2$ the larger the decay. 
	
	\begin{figure*}[htbp]
		\centering
		\includegraphics[width=0.9\textwidth]{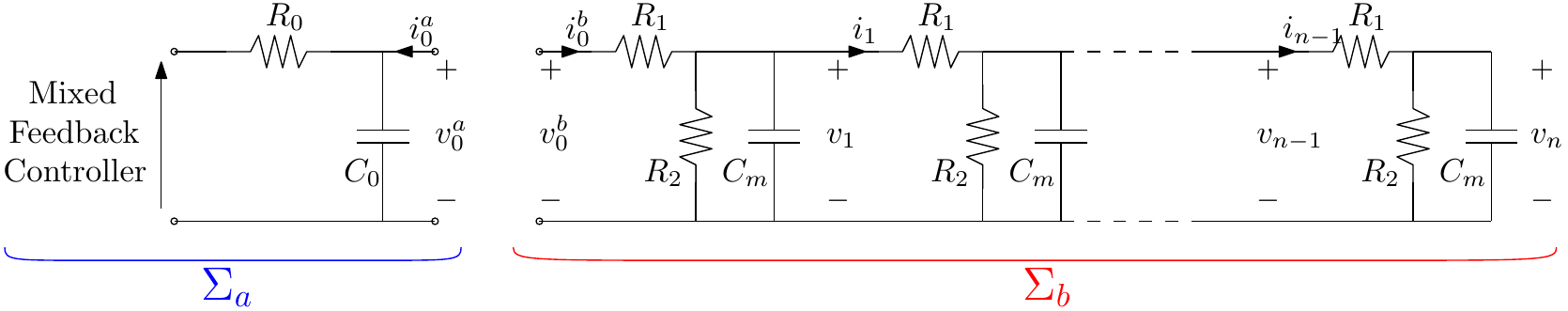}
		\caption{The mixed feedback closed loop $\Sigma_a$ interconnected
		with a passive network $\Sigma_b$.}
		\label{fig:RCnetwork}
	\end{figure*}
	
	\begin{figure}[htbp]
		\centering
		\includegraphics[width=0.9\columnwidth]{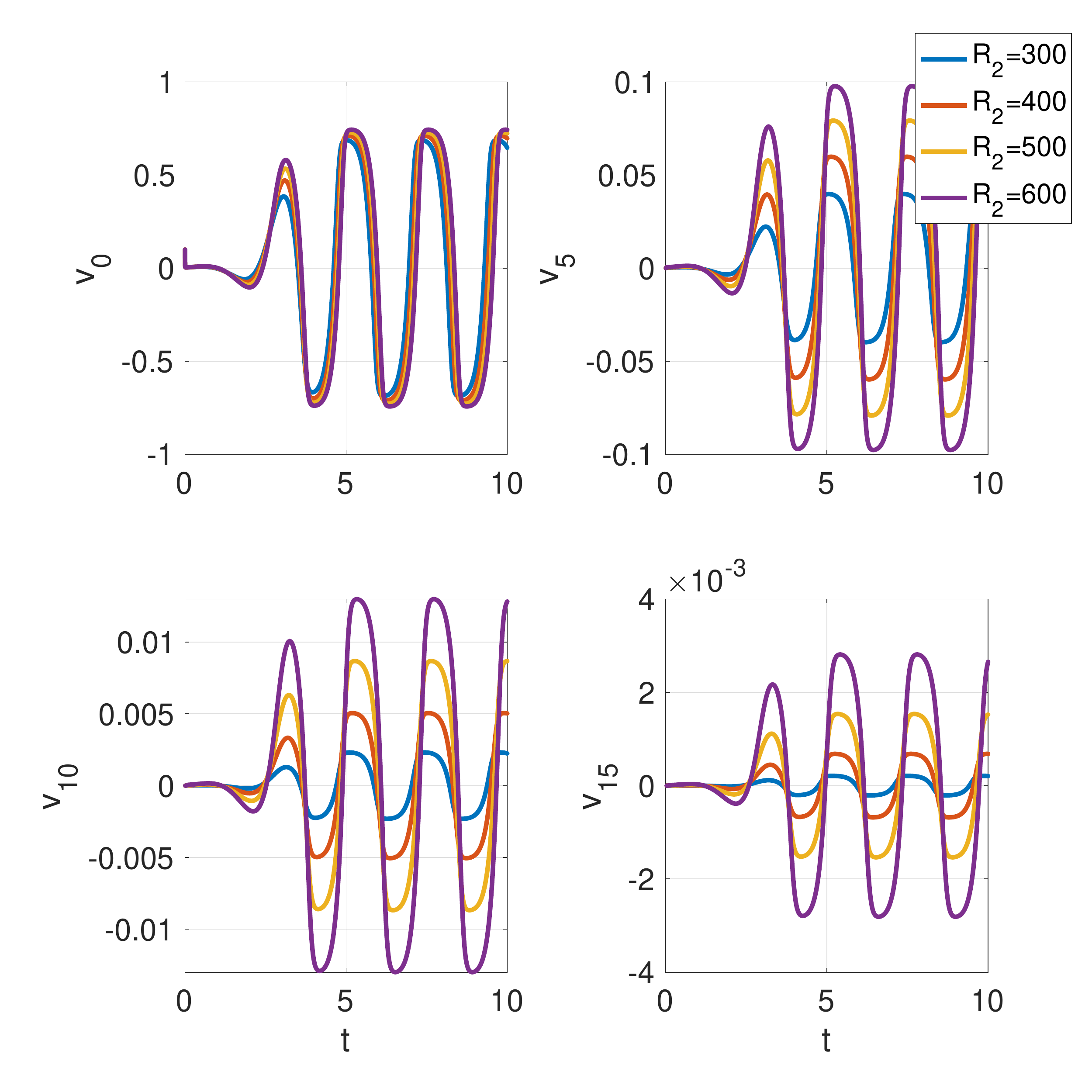} \vspace{-5mm}
		\caption{Sampled voltages of the interconnected system given by $\Sigma_a$ and $\Sigma_b$  for $R_2\in\{300,400,500,600\}$.}
		\label{fig:RCnetwork_simu}
	\end{figure}

	\section{CONCLUSIONS}
	
	We have studied the mixed-feedback controller as a robust generator of endogenous oscillations in closed loop.
	We have shown that the balance between fast positive and slow negative feedback is crucial to achieve stable oscillations.
	Grounded on dominance theory, we have derived sufficient conditions on the feedback gain $k$ and on
	the balance $\beta$ to achieve stable and robust oscillations in closed loop. These conditions have also been 
	extended to state-feedback design. Using LMIs, we have derived systematic design procedures to
	guarantee robust oscillations to bounded dynamic uncertainties and for passive interconnections. 
	Our design shows strong analogies with classical feedback design for stability. 
	This suggests a number of possible extensions, like the use of weighting functions for robustness, 
	or the characterization of mixed controllers based on output feedback, via state estimation.
	The results of the paper provide a theoretical justification to the observations from system biology and neuroscience that 
	mixed feedback is a fundamental mechanism for robust oscillations. 
	Our mixed feedback controller is limited to a single nonlinearity (saturation). 
	This leaves open questions of scalability to larger systems with several nonlinearities and of implementation into simple hardware. 
	This will be the object of future research.
	 

	\bibliography{Introduction_Ref}
	\bibliographystyle{ieeetr}
	
	\begin{IEEEbiography}[{\includegraphics[width=1in,height=1.25in,clip,keepaspectratio]{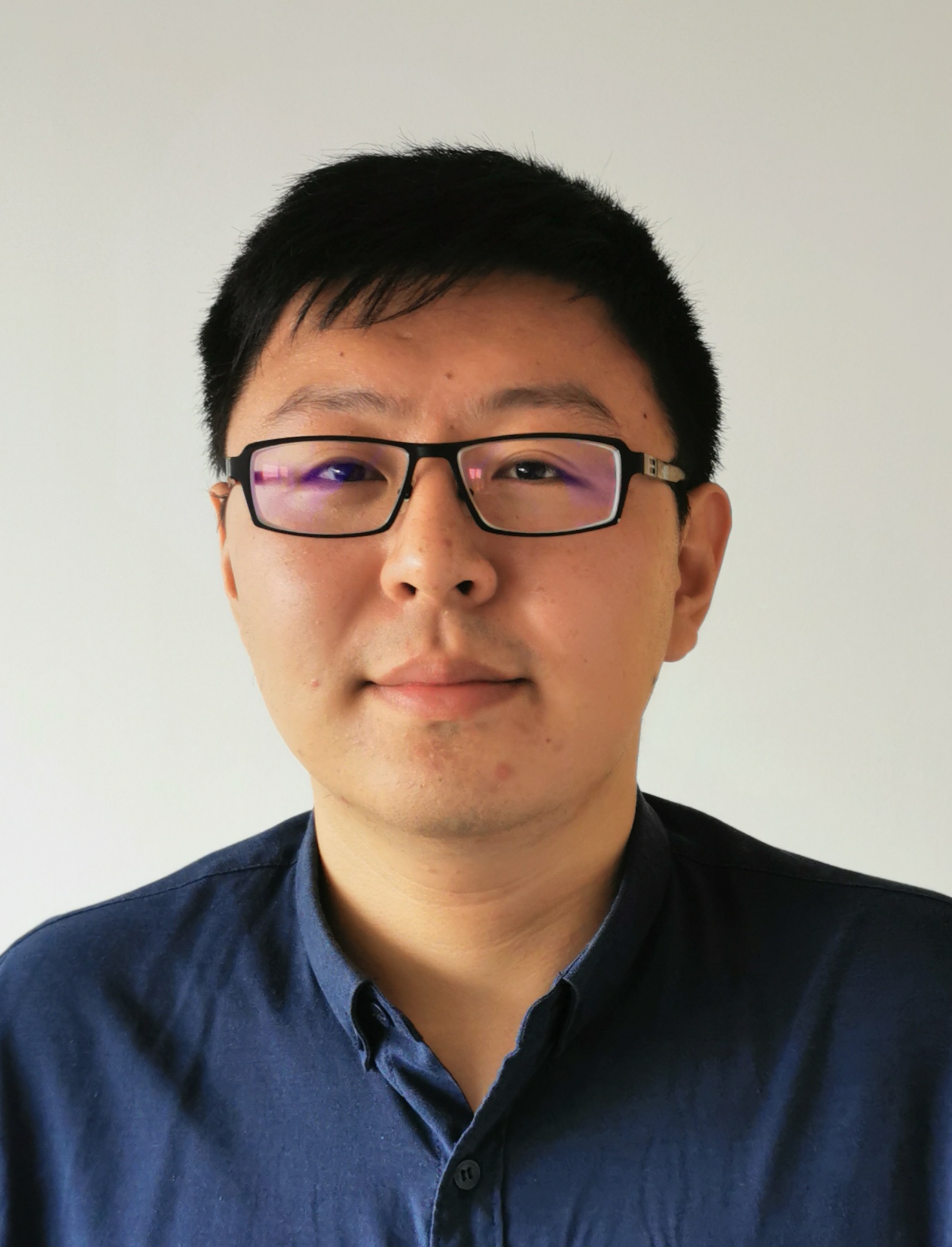}}]{Weiming Che} received the B.A. and M.Eng. degrees in Information Engineering both from the University of Cambridge, UK, in 2018. He is currently pursuing the Ph.D. degree with the University of Cambridge, UK. His research interests include the analysis and control of nonlinear system, specifically  bistable switches and oscillators.
	\end{IEEEbiography}

	\begin{IEEEbiography}
	[{\includegraphics[width=1in,height=1.25in,clip,keepaspectratio]{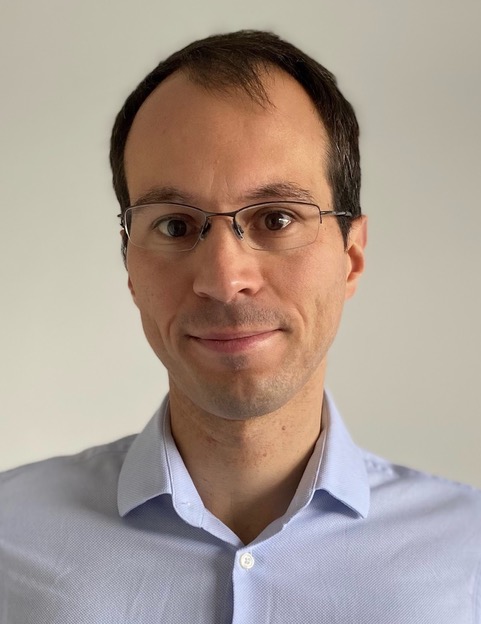}}]{Fulvio Forni} received the Ph.D. degree in computer science and control engineering from the University of Rome Tor Vergata, Rome, Italy, in 2010. In 2008–2009, he held visiting positions with the LFCS, University of Edinburgh, U.K. and with the CCDC of the University of California Santa Barbara, USA. In 2011–2015, he held a post-doctoral position with the University of Liege, Belgium (FNRS). He is currently Associate Professor in the Department of Engineering, University of Cambridge, UK. Dr. Forni was a recipient of the 2020 IEEE CSS George S. Axelby Outstanding Paper Award.
	\end{IEEEbiography}
	
\end{document}